\documentclass[a4paper]{lipics-v2019}
\usepackage[T1]{fontenc}
\usepackage{amssymb}
\usepackage{amsmath}
\usepackage{xspace}
\usepackage{xcolor} 
\usepackage{proba} 
\usepackage{mathtools}
\usepackage{hyperref} 
\usepackage{boxedminipage}
\usepackage{float}
\usepackage{bm}
\usepackage{nicefrac}
\usepackage{bbm}

\usepackage{tikz}
\usetikzlibrary{shapes.geometric}
\usetikzlibrary{arrows}
\usetikzlibrary{positioning}
\usetikzlibrary{fit}
\usetikzlibrary{calc}
\usetikzlibrary{backgrounds}
\usetikzlibrary{shapes}
\usetikzlibrary{patterns}
\usetikzlibrary{matrix}
\usetikzlibrary{intersections}
\usetikzlibrary{decorations}

\usepackage[strings]{underscore} 

\hypersetup{pdfstartview=FitH,colorlinks,linkcolor=blue,filecolor=blue,citecolor=blue,urlcolor=blue}

\tikzset{
	itria/.style={
		draw,dashed,shape border uses incircle,
		isosceles triangle,shape border rotate=90,yshift=-3.2cm}
}
\newcommand{\tuple}[1]{\ensuremath{{\left\langle{#1}\right\rangle}}\xspace}

\newenvironment{boxedalgo}
  {\begin{center}\begin{boxedminipage}{1\textwidth}}
  {\end{boxedminipage}\end{center}}

\providecommand{\ie}{\text{i.e.}\xspace}
\providecommand{\etal}{\text{et al.}\xspace}
\providecommand{\wrt}{\text{w.r.t.}\xspace}
\newcommand{\naive}{\text{na\"ive}\xspace}

\providecommand{\namedref}[2]{\hyperref[#2]{#1~\ref*{#2}}\xspace}
\providecommand{\lemmaref}[1]{\namedref{Lemma}{lem:#1}}
\providecommand{\theoremref}[1]{\namedref{Theorem}{thm:#1}}
\providecommand{\corollaryref}[1]{\namedref{Corollary}{corol:#1}}
\providecommand{\sectionref}[1]{\namedref{Section}{sec:#1}}
\providecommand{\appendixref}[1]{\namedref{Appendix}{app:#1}}
\providecommand{\figureref}[1]{\namedref{Fig.}{fig:#1}}
\providecommand{\claimref}[1]{\namedref{Claim}{clm:#1}}
\providecommand{\zo}{\ensuremath{{\{0,1\}}}\xspace}

\providecommand{\poly}{\ensuremath{\mathrm{poly}}\xspace}

\providecommand{\eps}[0]{\ensuremath{\varepsilon}}
\let\epsilon\eps
\let\leq\leqslant

\let\geq\geqslant

\newcommand{\cX}{\ensuremath{{\mathcal X}}\xspace}

\providecommand{\defeq}[0]{\ensuremath{\coloneqq}\xspace}
\providecommand{\eqdef}[0]{\ensuremath{\eqqcolon}\xspace}
\DeclarePairedDelimiter\abs{\lvert}{\rvert}

\DeclareMathOperator*{\argmax}{arg\,max}
\providecommand{\p}[1]{\ensuremath{^{{\left(#1\right)}}}\xspace}

\tikzstyle{vertex}=[circle, draw, inner sep=0pt, minimum size=6pt]
\providecommand{\vertex}{\node[vertex]}
\tikzset{
	buffer/.style={
		draw,dashed,shape border uses incircle,
		isosceles triangle,shape border rotate=90,yshift=-2.4cm,xshift=0.04cm, inner sep=0pt}
}

\usepackage{microtype}

\title{%
	Estimating Gaps in Martingales and Applications to Coin-Tossing: Constructions \& Hardness
}

\author{Hamidreza Amini Khorasgani}{Department of Computer Science, Purdue University, IN, USA}{haminikh@purdue.edu}{}{}

\author{Hemanta K. Maji}{Department of Computer Science, Purdue University, IN, USA}{hmaji@purdue.edu}{}{}

\author{Tamalika Mukherjee}{Department of Computer Science, Purdue University, IN, USA}{tmukherj@purdue.edu}{}{}

\authorrunning{H.\, Amini Khorasgani, H.\,K. Maji and T.\,Mukherjee}

\funding{The research effort is supported in part by an NSF CRII Award CNS--1566499, an NSF SMALL Award CNS--1618822, the IARPA HECTOR project, MITRE Innovation Program Academic Cybersecurity Research Award, a Purdue Research Foundation (PRF) Award, and  The Center for Science of Information, an NSF Science and Technology Center, Cooperative Agreement CCF--0939370.}

\Copyright{Hamidreza Amini Khorasgani, Hemanta K. Maji and Tamalika Mukherjee}

\ccsdesc[500]{Mathematics of computing~Markov processes}
\ccsdesc[500]{Security and privacy~Information-theoretic techniques}
\ccsdesc[300]{Security and privacy~Mathematical foundations of cryptography}

\keywords{Discrete-time Martingale, Coin-tossing and Dice-rolling Protocols, Discrete Control Processes, Fair Computation, Black-box Separation}

\nolinenumbers 
\hideLIPIcs

\begin{document}
	
	\maketitle
	{%
		\thispagestyle{empty}%
		\begin{abstract}
Consider the representative task of designing a distributed coin-tossing protocol for $n$ processors such that the probability of heads is $X_0\in[0,1]$. 
This protocol should be robust to an adversary who can reset one processor to change the distribution of the final outcome. 
For $X_0=1/2$, in the information-theoretic setting, no adversary can deviate the probability of the outcome of the well-known Blum's ``majority protocol'' by more than $\frac1{\sqrt{2\pi n}}$, \ie, it is $\frac1{\sqrt{2\pi n}}$ insecure.

In this paper, we study discrete-time martingales $(X_0,X_1,\dotsc,X_n)$ such that $X_i\in[0,1]$, for all $i\in\{0,\dotsc,n\}$, and $X_n\in\zo$. 
These martingales are commonplace in modeling stochastic processes like coin-tossing protocols in the information-theoretic setting mentioned above. 
In particular, for any $X_0\in[0,1]$, we construct martingales that yield  $\frac12\sqrt{\frac{X_0(1-X_0)}{n}}$ insecure coin-tossing protocols. 
For $X_0=1/2$, our protocol requires only 40\% of the processors to achieve the same security as the majority protocol.

The technical heart of our paper is a new inductive technique that uses geometric transformations to precisely account for the large gaps in these martingales. 
For any $X_0\in[0,1]$, we show that there exists a stopping time $\tau$ such that
  $$\EX{\abs{X_\tau-X_{\tau-1}}}\geq \frac2{\sqrt{2n-1}}\cdot X_0(1-X_0)$$
The inductive technique simultaneously constructs martingales that demonstrate the optimality of our bound, \ie, a martingale where the gap corresponding to any stopping time is small. 
In particular, we construct optimal martingales such that {\em any} stopping time $\tau$ has 
  $$\EX{\abs{X_\tau-X_{\tau-1}}}\leq \frac1{\sqrt{n}}\cdot \sqrt{X_0(1-X_0)}$$ 
Our lower-bound holds for all $X_0\in[0,1]$; while the previous bound of Cleve and Impagliazzo (1993) exists only for positive constant $X_0$. 
Conceptually, our approach only employs elementary techniques to analyze these martingales and entirely circumvents the complex probabilistic tools inherent to the approaches of Cleve and Impagliazzo (1993) and Beimel, Haitner, Makriyannis, and Omri (2018).

By appropriately restricting the set of possible stopping-times, we present representative applications to constructing distributed coin-tossing/dice-rolling protocols, discrete control processes, fail-stop attacking coin-tossing/dice-rolling protocols, and black-box separations.
\end{abstract}%
		\newpage%
		\thispagestyle{empty}%
		\tableofcontents%
	}
	\newpage%
	
	\setcounter{page}{1}
	\section{Introduction}
\label{sec:intro} 

{\bfseries A Representative Motivating Application. }
Consider a distributed protocol for $n$ processors to toss a coin, where a processor $i$ broadcasts her message in round $i$. 
At the end of the protocol, all processors reconstruct the common outcome from the public transcript. 
When all processors are honest, the probability of the final outcome being 1 is $X_0$ and the probability of the final outcome being 0 is $1-X_0$, \ie, the final outcome is a {\em bias-$X_0$ coin}. 
Suppose there is an adversary who can (adaptively) choose to {\em restart} one of the processors after seeing her message (\ie, the {\em strong adaptive} corruptions model introduced by Goldwasser, Kalai, and Park~\cite{ICALP:GolKalPar15}); otherwise her presence is innocuous. 
Our objective is to design bias-$X_0$ coin-tossing protocols such that the adversary cannot change the distribution of the final outcomes significantly. 

{\em The Majority Protocol.} 
Against computationally unbounded adversaries, (essentially) the only known protocol is the well-known majority protocol~\cite{STOC:Blum83,ABCGM85,STOC:Cleve86} for $X_0=1/2$.
The majority protocol requests one uniformly random bit from each processor and the final outcome is the majority of these $n$ bits. 
An adversary can alter the probability of the final outcome being 1 by $\frac1{\sqrt{2\pi n}}$, \ie, the majority protocol is $\frac1{\sqrt{2\pi n}}$ insecure. 

{\em Our New Protocol.}
We shall prove a general martingale result in this paper that yields the following result as a corollary. 
For any $X_0\in[0,1]$, there exists an $n$-bit bias-$X_0$ coin-tossing protocol in the information-theoretic setting that is $\frac12\sqrt{\frac{X_0(1-X_0)}{n}}$ insecure.
In particular, for $X_0=1/2$, our protocol uses only 625 processors to reduce the insecurity to, say,  1\%; while the majority protocol requires 1592 processors.

{\bfseries General Formal Framework: Martingales.} 
Martingales are natural models for several stochastic processes. 
Intuitively, martingales correspond to a gradual release of information about an event. 
A priori, we know that the probability of the event is $X_0$. 
For instance, in a distributed $n$-party coin-tossing protocol the outcome being $1$ is the event of interest.   

A discrete-time martingale $(X_0,X_1,\dotsc,X_n)$ represents the gradual release of information about the event over $n$ time-steps.\footnote{%
  For the introduction, we do not explicitly mention the underlying filtration for brevity. 
  The proofs, however, clearly mention the associated filtrations.%
}
For intuition, we can assume that $X_i$ represents the probability that the outcome of the coin-tossing protocol is $1$ after the first $i$ parties have broadcast their messages. 
Martingales have the unique property that if one computes the expected value of $X_j$, for $j>i$, at the end of time-step $i$, it is identical to the value of $X_i$. 
In this paper we shall consider martingales where, at the end of time-step $n$, we know for sure whether the event of interest has occurred or not.
That is, we have $X_n\in\zo$. 

A {\em stopping time} $\tau$ represents a time step $\in \{1,2,\dotsc,n\}$ where we stop the evolution of the martingale. 
The test of whether to stop the martingale at time-step $i$ is a function only of the information revealed so far. 
Furthermore, this stopping time need {\em not} be a constant. 
That is, for example, different transcripts of the coin-tossing protocol potentially have different stopping times.

{\bfseries Our Martingale Problem Statement.} 
The inspiration of our approach is best motivated using a two-player game between, namely, the {\em martingale designer} and the {\em adversary}. 
Fix $n$ and $X_0$. 
The martingale designer presents a martingale $\cX = (X_0,X_1,\dotsc,X_n)$ to the adversary and the adversary finds a stopping time $\tau$ that maximizes the following quantity. 
  $$\EX{\abs{X_\tau-X_{\tau-1}}}$$
Intuitively, the adversary demonstrates the most severe {\em susceptibility} of the martingale by presenting the corresponding stopping time $\tau$ as a witness. 
The martingale designer's objective is to design martingales that have less susceptibility. 
Our paper uses a geometric approach to inductively provide tight bounds on the least susceptibility of martingales for all $n\geq 1$ and $X_0\in[0,1]$, that is, the following quantity.
$$C_n(X_0) \defeq \inf_{\cX} \; \sup_\tau \; \EX{\abs{X_\tau-X_{\tau-1}}}$$
This precise study of $C_n(X_0)$, for general $X_0\in[0,1]$, is motivated by natural applications in discrete process control as illustrated by the representative motivating problem. 
This paper, for representative applications of our results, considers $n$-processor distributed protocols and 2-party $n$-round protocols.
The stopping time witnessing the highest susceptibility shall translate into appropriate adversarial strategies. 
These adversarial strategies shall imply hardness of computation results. 

\subsection{Our Contributions}
\label{sec:contrib} 

We prove the following general martingale theorem.
\begin{theorem}
	\label{thm:gap-main} 
	Let $(X_0,X_1,\dotsc,X_n)$ be a discrete-time martingale such that $X_i\in[0,1]$, for all $i\in\{1,\dotsc,n\}$, and
	$X_n\in\{0,1\}$.
	Then, the following bound holds.
	$$\sup_{\text{stopping time }\tau} \EX{\abs{X_\tau-X_{\tau-1}}} \geq C_n(X_0),$$
	where $C_1(X) = 2X(1-X)$, and, for $n>1$, we obtain $C_n$ from $C_{n-1}$ recursively using the geometric transformation defined in \figureref{transform-def}.
	
	Furthermore, for all $n\geq1$ and $X_0\in[0,1]$, there exists a martingale $(X_0,\dotsc,X_n)$ (\wrt to the coordinate exposure filtration for $\zo^n$) such that for any stopping time $\tau$, it has $\EX{\abs{X_\tau-X_{\tau-1}}} = C_n(X_0)$. 
\end{theorem} 

Intuitively, given a martingale, an adversary can identify a stopping time where the expected gap in the martingale is at least $C_n(X_0)$.
Moreover, there exists a martingale that realizes the lower-bound in the tightest manner, \ie, all stopping times $\tau$ have identical susceptibility. 

Next, we estimate the value of the function $C_n(X)$. 
\begin{lemma}
	\label{lem:gap-lower-upper}
	For $n\geq 1$ and $X\in[0,1]$, we have 
	$$\frac2{\sqrt{2n-1}}X(1-X) \eqdef L_n(X) \leq C_n(X) \leq U_n(X)\defeq \frac1{\sqrt{n}}\sqrt{X(1-X)}$$
\end{lemma}
As a representative example, consider the case of $n=3$ and $X_0=1/2$. 
\figureref{3-majority} presents the martingale corresponding to the 3-round majority protocol and highlights the stopping time witnessing the susceptibility of 0.3750. 
\figureref{3-opt} presents the optimal 3-round coin-tossing protocol's martingale that has susceptibility of 0.2407. 

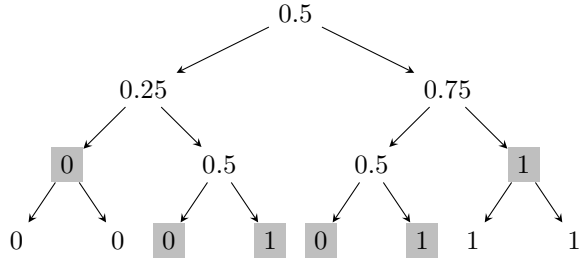
\begin{figure}
\begin{center}
\begin{tikzpicture}
[->,>=stealth,level/.style={sibling distance = 4cm/#1,
	level distance = 1cm}
] 
\node [] {0.5}
child{ node [] {0.25} 
	child{ node [fill=gray!50!] {0} 
		child{ node [] {0}}
		child{ node [] {0}}
	}
	child{ node [] {0.5}
		child{ node [fill=gray!50!] {0}  
		}
		child{ node [fill=gray!50!] {1} 
		}
	}                          
}
child{ node [] {0.75}
	child{ node [] {0.5} 
		child{ node [fill=gray!50!] {0} 
		}
		child{ node [fill=gray!50!] {1} 
		}
	}
	child{ node [fill=gray!50!] {1}
		child{ node [] {1}}
		child{ node [] {1}}
	}
}
; 
\end{tikzpicture}
\end{center}
\caption{Majority Protocol Tree of depth three. 
  The optimal score in the majority tree of depth three is $0.3750$ and the corresponding stopping time is highlighted in gray.%
}
\label{fig:3-majority}
\end{figure}
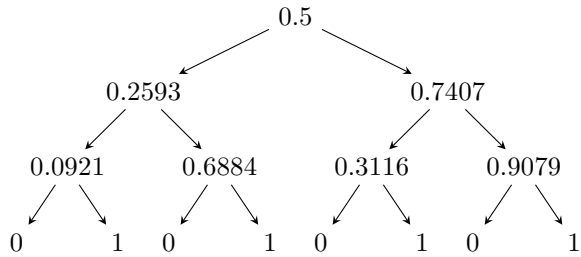
\begin{figure}
\begin{center}
\begin{tikzpicture}[->,>=stealth,level/.style={sibling distance = 4cm/#1,
	level distance = 1cm}] 
\node [] {0.5}
child{ node [] {0.2593} 
	child{ node [] {0.0921} 
		child{ node [] {0}}
		child{ node [] {1}}
	}
	child{ node [] {0.6884}
		child{ node [] {0}}
		child{ node [] {1}}
	}                            
}
child{ node [] {0.7407}
	child{ node [] {0.3116} 
		child{ node [] {0}}
		child{ node [] {1}}
	}
	child{ node [] {0.9079}
		child{ node [] {0}}
		child{ node [] {1}}
	}
}
; 
\end{tikzpicture}
\end{center}
\caption{Optimal depth-3 protocol tree for $X_0=1/2$. 
The optimal score is $0.2407$. 
Observe that any stopping time achieves this score.}
\label{fig:3-opt}
\end{figure} 

In the sequel, we highlight applications of \theoremref{gap-main} to protocol constructions and hardness of computation results using these estimates. 

\begin{remark}[Protocol Constructions]
	The optimal martingales naturally translate into $n$-bit distributed coin-tossing and multi-faceted dice rolling protocols. 
	\begin{enumerate}
		\item \corollaryref{fair}: 
		For all $X_0\in[0,1]$, there exists an $n$-bit distributed bias-$X_0$ coin-tossing protocol for $n$ processors with the following security guarantee. 
		Any (computationally unbounded) adversary who follows the protocol honestly and resets at most one of the processors during the execution of the protocol can change the probability of an outcome by at most $\frac1{2\sqrt n}\sqrt{X_0(1-X_0)}$. 
	\end{enumerate}
\end{remark}

\begin{remark}[Hardness of Computation Results]
	The lower-bound on the maximum susceptibility helps demonstrate hardness of computation results. 
	For $X_0=1/2$, Cleve and Impagliazzo~\cite{Cleve93martingales} proved that one encounters $\abs{X_\tau-X_{\tau-1}}\geq \frac1{32\sqrt n}$ with probability $\frac15$. 
	In other words, their bound guarantees that the expected gap in the martingale is at least $\frac1{160\sqrt n}$, which is significantly smaller than our bound $\frac1{2\sqrt{2 n}}$. 
	Hardness of computation results relying on \cite{Cleve93martingales} (and its extensions) work only for constant $0<X_0<1$.%
	\footnote{%
		Cleve and Impagliazzo set their problem as an optimization problem that trades off two conflicting objective functions. 
		These objective functions have exponential dependence on $X_0(1-X_0)$. 
		Consequently, if $X_0=1/\poly(n)$ or $X_0=1-1/\poly(n)$, then their lower bounds are extremely weak.
	}
	However, our lower-bound holds for all $X_0\in[0,1]$; for example, even when $1/\poly(n) \leq X_0  \leq 1-1/\poly(n)$. 
	Consequently, we extend existing hardness of computation results using our more general lower-bound. 
	\begin{enumerate}
		\item \theoremref{unfairness} extends the fail-stop attack of \cite{Cleve93martingales} on 2-party bias-$X_0$ coin-tossing protocols (in the information-theoretic commitment hybrid). 
		For any $X_0\in[0,1]$, a fail-stop adversary can change the probability of the final outcome of any 2-party bias-$X_0$ coin-tossing protocol by $\geq \frac{\sqrt 2}{12\sqrt{n+1}}X_0(1-X_0)$. 
		This result is useful to demonstrate black-box separations results. 
		\item \corollaryref{bb-sep} extends the black-box separation results of \cite{TCC:DLMM11,TCC:HaiOmrZar13,TCC:DacMahMal14} separating (appropriate restrictions of) 2-party bias-$X_0$ coin tossing protocols from one-way functions. 
		We illustrate a representative new result that follows as a consequence of \corollaryref{bb-sep}.  
		For constant $X_0\in(0,1)$, \cite{TCC:DLMM11,TCC:HaiOmrZar13,TCC:DacMahMal14} rely on (the extensions of) \cite{Cleve93martingales} to show that it is highly unlikely that there exist 2-party bias-$X_0$ coin tossing protocols using one-way functions in a black-box manner achieving  $o(1/\sqrt n)$ {\em unfairness}~\cite{EC:GorKat10}. 
		Note that when $X_0=1/n$, there are secure 2-party coin tossing protocols with  $1/2n$ unfairness (based on \corollaryref{fair}) even in the information-theoretic setting. 
		Previous results cannot determine the limits to the unfairness of 2-party bias-$1/n$ fair coin-tossing protocols that use one-way functions in a black-box manner. 
		Our black-box separation result (refer to \corollaryref{bb-sep}) implies that it is highly unlikely to construct bias-$1/n$ coin using one-way functions in a black-box manner with $< \frac{\sqrt2}{12\cdot n^{3/2}}$ unfairness. 
		\item \corollaryref{control} and \corollaryref{control2} extend Cleve and Impagliazzo's~\cite{Cleve93martingales} result on influencing discrete control processes to arbitrary $X_0\in[0,1]$. 
	\end{enumerate}
\end{remark}

\subsection{Prior Approaches to the General Martingale Problem} 
\label{sec:prior} 

Azuma-Hoeffding inequality~\cite{Azuma1967,Hoeffding63} states that if $\abs{X_i-X_{i-1}} = o(1/\sqrt n)$, for all $i\in\{1,\dotsc,n\}$, then, essentially, $\abs{X_n-X_0}=o(1)$ with probability 1.
That is, the final information $X_n$ remains close to the a priori information $X_0$. 
However, in our problem statement, we have $X_n\in\zo$. 
In particular, this constraint implies that the final information $X_n$ is significantly different from the a priori information $X_0$.
So, the initial constraint ``for all $i\in\{1,\dotsc,n\}$ we have $\abs{X_i-X_{i-1}} = o(1/\sqrt n)$'' must be violated. 
What is the probability of this violation? 

For $X_0=1/2$, Cleve and Impagliazzo~\cite{Cleve93martingales} proved that there exists a  round $i$ such that $\abs{X_i-X_{i-1}}\geq \frac1{32\sqrt n}$ with probability $1/5$. 
We emphasize that the round $i$ is a random variable and not a constant. 
However, the definition of the ``big jump'' and the ``probability to encounter big jumps'' both are exponentially small function of $X_0$.
So, the approach of Cleve and Impagliazzo is only applicable to constant $X_0\in(0,1)$.
Recently, in an independent work, Beimel~\etal~\cite{BeimelHMO17} demonstrate an identical bound for {\em weak martingales} (that have some additional properties), which is used to model multi-party coin-tossing protocols. 

For the upper-bound, on the other hand, Doob's martingale corresponding to the majority protocol is the only known martingale for $X_0=1/2$ with a small {\em maximum susceptibility}. 
In general, to achieve arbitrary $X_0\in[0,1]$, one considers coin tossing protocols where the outcome is $1$ if the total number of heads in $n$ uniformly random coins surpasses an appropriate threshold.

	\section{Preliminaries}
\label{sec:prelim}

We denote the {\em arithmetic mean} of two numbers $x$ and $y$ as $\mathrm{A.M.}(x,y)\defeq (x+y)/{2}$. 
The {\em geometric mean} of these two numbers is denoted by $\mathrm{G.M.}(x,y)\defeq \sqrt{x\cdot y}$ and their {\em harmonic mean} is denoted by $\mathrm{H.M.}(x,y)\defeq \left(\left(x^{-1}+y^{-1}\right)/2\right)^{-1}=2xy/(x+y)$. 

{\bfseries Martingales and Related Definitions.}
The {\em conditional expectation} of a random variable $X$ with respect to an event $\mathcal{E}$ denoted by  $\EX{X|\mathcal{E}}$, is defined as $\EX{X\cdot\1{\mathcal E}}/\probX{\mathcal{E}}$. 
For a discrete random variable $Y$, the conditional expectation of $X$ with respect to $Y$, denoted by $\EX{X|Y}$, is a random variable that takes value $\EX{X|Y=y}$ with probability $\probX{Y=y}$, where $\EX{X|Y=y}$ denotes the conditional expectation of $X$ with respect to the event $\{\omega\in \Omega|Y(\omega)=y\}$. 

Let $\Omega=\Omega_1\times\Omega_2\times\dotsi\times\Omega_n$ denote a sample space and $(E_1,E_2,\dotsc,E_n)$ be a joint distribution defined over $\Omega$ such that 
for each $i\in\{1,\dotsc,n\}$, $E_i$ is a random variable over $\Omega_i$. 
Let $X=\{X_i\}_{i=0}^{n}$ be a sequence of random variables defined over $\Omega$. 
We say that $X_j$ is $E_1,\dotsc,E_j$ measurable if there exists a function $g_j\colon \Omega_1\times\Omega_2\times\dotsi\times\Omega_j\to\mathbb{R} $ such that $X_j=g_j(E_1,\dotsc,E_j)$.
Let $X=\{X_i\}_{i=0}^{n}$ be a discrete-time martingale sequence with respect to the sequence $E=\{E_i\}_{i=1}^{n}$.
This statement implies that for each $i\in\{0,1,\dotsc,n\}$, we have 
$$\EX{X_{i+1}|E_1,E_2,\dotsc,E_i}=X_i$$
Note that the definition of martingale implies $X_i$ to be $E_1,\dotsc,E_i$ measurable for each $i\in\{1,\dotsc,n\}$ and $X_0$ to be constant.  
In the sequel, we shall use $\{ X=\{X_i\}_{i=0}^{n},E=\{E_i\}^n_{i=1} \}$ to denote a martingale sequence where for each $i=1,\dotsc,n$, $X_i\in [0,1]$, and $X_n\in \{0,1\}$. 
However, for brevity, we use $\left(X_0,X_1,\dotsc,X_n\right)$ to denote a martingale. 
Given a function $f\colon\Omega_1\times\Omega_2\times\dotsi\times\Omega_n\to\mathbb{R}$, if we define the random variable $Z_i\defeq \EX{f(E_1,\dotsc,E_n)|E_1,\dotsc,E_i}$, for each $i\in \{0,1,\dotsc,n\}$, then the sequence $Z=\{Z_i\}_{i=0}^{n}$ is a martingale with respect to $\{E_i\}_{i=1}^{n}$. 
This martingale is called the {\em Doob's martingale}.

The random variable $\tau\colon\Omega \rightarrow \{0,1,\dotsc, n\}$ is called a stopping time if for each $k\in\{1,2,\dotsc,n\}$, the occurrence or non-occurrence of the event $\{\tau\leq k\} \defeq \{\omega\in \Omega| \tau(\omega)\leq k\}$ depends only on the values of random variables $E_1,E_2,\dotsc,E_k$.
Equivalently, the random variable $\1{\tau \leq k}$ is $E_1,\dotsc, E_k$ measurable. 
Let $\mathcal{S}(X,E)$ denote the set of all stopping time random variables over the  martingale sequence  $\{ X=\{X_i\}_{i=0}^{n},E=\{E_i\}^n_{i=1} \}$. 
For $ \ell \in \{1,2\}$, we define the {\em score} of a martingale sequence $(X,E)$ with respect to a stopping time $\tau$ in the $L_\ell$-norm as the following quantity.
$$\mathrm{score}_\ell(X,E,\tau) \defeq \EX{ \abs{X_\tau - X_{\tau-1}}^\ell }$$
We define the {\em max stopping time} as the stopping time that maximizes the score 
$$\tau_{\max}(X,E,\ell) \defeq \argmax_{\tau\in \mathcal{S}(X,E)} \mathrm{score}_\ell(X,E,\tau),$$ 
and the (corresponding) {\em $\mathrm{max\textnormal{-}score}$} as 
$$\mathrm{max\text{-}score}_\ell(X,E)\defeq\EX{|X_{\tau_{\max}}-X_{\tau_{\max-1}}|^\ell}$$

Let $A_n(x^{*})$ denote the set of all discrete time martingales $\{X=\{X_i\}_{i=0}^n,E=\{E_i\}_{i=1}^n\}$ such that $X_0=x^{*}$ and $X_n \in \{0,1\}$. 
We define {\em optimal score} as 
$$\mathrm{opt}_n(x^{*},\ell) \defeq \inf_{(X,E)\in A_n(x^{*})}\mathrm{max\textnormal{-}score}_\ell(X,E)$$
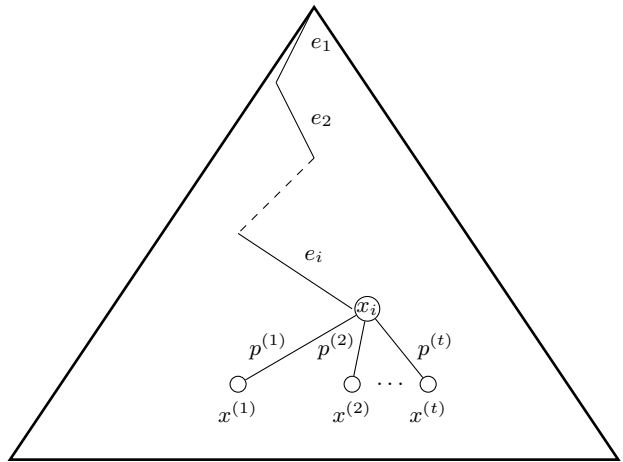
\begin{figure}
\begin{center}
\begin{tikzpicture}[]\footnotesize 
	\coordinate (A) at (0,6);
	\coordinate (B) at ( 4,0);
	\coordinate (C) at (-4,0);
	\draw[name path=b1, line width=1pt] (A) -- (B) -- (C) -- cycle;
	\draw[-] (A) -- (-0.5,5) node[right] {}; 
	\draw[-] (-0.5,5) -- (0,4) node[right] {};
	\draw[-,dashed] (0,4) -- (-1,3) node[right] {};
	\draw[-] (-1,3) -- (0.5,2) node[right] {};	
	\node [circle ,draw,inner sep=0pt, minimum size=6pt] (v) at (0.7, 2) {$x_i$};
	\node [circle,draw,inner sep=0pt, minimum size=6pt] (z) at (-1, 1) [label=below:{$x^{(1)}$}] {};
	\node [circle,draw,inner sep=0pt, minimum size=6pt] (x) at (0.5, 1) [label=below:{$x^{(2)}$}] {};
	\vertex (y) at (1,1) [draw=none]{$\ldots$};
	\node [circle,draw,inner sep=0pt, minimum size=6pt] (w) at (1.5, 1) [label=below:{$x^{(t)} $}] {};
	\node at (0.1,5.5){$e_1$};
	\node at (0.1,4.5){$e_2$};
	\node at (0,2.7){$e_i$};
	\node at (-0.6,1.5){$p^{(1)}$};
	\node at (0.3,1.5){$p^{(2)}$};
	\node at (1.6,1.5){$p^{(t)}$};
	\draw[-] (v) -- (x) node[right] {};
	\draw[-] (v) -- (w) node[right] {};
	\draw[-] (v) -- (z) node[right] {};
\end{tikzpicture}
	\caption{Interpreting a general martingale as a tree.}
	\vspace*{-0.5cm}%
	\label{fig:prelim-mart-tree}
\end{center}
\end{figure}

{\bfseries Representing a Martingale as a Tree.}	
We interpret a discrete time martingale sequence $X=\{X_i\}_{i=0}^{n}$ defined over a sample space $\Omega=\Omega_1\times\dotsi\times\Omega_n$ as a tree of depth $n$ (see \figureref{prelim-mart-tree}). 
For $i=0,\dotsc,n$, any node at depth $i$ has $\abs{\Omega_{i+1}}$ children. 
In fact, for each $i$, the edge between a node at depth $i$ and a child at depth $(i+1)$ corresponds to a possible outcome that $E_{i+1}$ can take from the set $\Omega_{i+1}=\{x\p1,\dotsc,x\p t\}$.

Each node $v$ at depth $i$ is represented by a unique path from root to $v$ like $(e_1,e_2,\dots,e_{i})$, which corresponds to the event $\{\omega\in \Omega|E_1(\omega)=e_1,\dots,E_{i}(\omega)=e_{i}\}$. 
Specifically, each path from root to a leaf in this tree, represents a unique outcome in the sample space $\Omega$.

Any subset of nodes in a tree that has the property that none of them is an ancestor of any other, is called an {\em anti-chain}. 
If we use our tree-based notation to represent a node $v$, \ie,  the sequence of edges $e_1,\dots,e_i$ corresponding to the path from root to $v$, then any prefix-free subset of nodes is an anti-chain. 
Any anti-chain that is not a proper subset of another anti-chain is called a {\em maximal anti-chain}. 
A stopping time in a martingale corresponds to a {\em unique} maximal anti-chain in the martingale tree.%

\noindent{\bfseries Geometric Definitions and Relations.}
Consider curves $C$  and $D$ defined by the zeroes of $Y=f(X)$ and $Y=g(X)$, respectively, where $X\in[0,1]$. 
We restrict to curves $C$ and $D$ such that each one of them have exactly one intersection with $X=x$, for any $x\in [0,1]$.
Refer to \figureref{prelim-good-ex} for intuition. 
Then, we say $C$ is \emph{above} $D$, represented by $C \succcurlyeq D $, if,  for each $x\in [0,1]$, we have $f(x) \geq g(x)$.
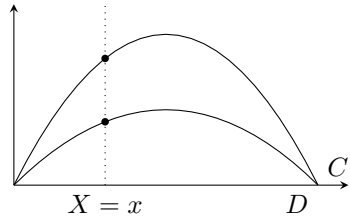
\begin{figure}
\begin{center}%
\begin{tikzpicture}[>=stealth,scale=4,domain=0:1]
	\coordinate (O) at (0,0); 
	\draw[->] (O) to ++(1.1,0); 
	\draw[->] (O) to ++(0,0.6); 
	
	\draw[] plot (\x,{2*\x*(1-\x)}) node[right,anchor=south west] {$C$};
	\draw[] plot (\x,{\x*(1-\x)}) node[right,anchor=north east] {$D$};
	
	\coordinate (X) at (0.3,0); 
	\draw[dotted] (X) to ++(0,.6); 
	\node [anchor=north] at (X) {$X=x$}; 
	
	\node[circle,inner sep=1pt,fill=black] at (0.3,.42) {};
	\node[circle,inner sep=1pt,fill=black] at (0.3,.21) {};
\end{tikzpicture} 
\caption{Intuition for a curve $C$ being above another curve $D$, represented by $C \succcurlyeq D $.}
\label{fig:prelim-good-ex}
\end{center}
\end{figure}

	\section{Large Gaps in Martingales: A Geometric Approach}
\label{sec:large-gap-l1}

This section presents a high-level overview of our proof strategy. 
In the sequel, we shall assume that we are working with discrete-time martingales $(X_0,X_1,\dotsc,X_n)$ such that $X_n\in\zo$.

Given a martingale $(X_0,\dotsc,X_n)$, its {\em susceptibility} is represented by the following quantity
  $$\sup_{\text{stopping time }\tau} \EX{\abs{X_\tau-X_{\tau-1}}}$$
Intuitively, if a martingale has high susceptibility, then it has a stopping time such that the gap in the martingale while encountering the stopping time is large. 
Our objective is to characterize the {\em least susceptibility} that a martingale $(X_0,\dotsc,X_n)$ can achieve. 
More formally, given $n$ and $X_0$, characterize
  $$C_n(X_0) \defeq \inf_{(X_0,\dotsc,X_n)} \sup_{\text{stopping time }\tau} \EX{\abs{X_\tau-X_{\tau-1}}}$$
Our approach is to proceed by induction on $n$ to exactly characterize the curve $C_n(X)$, and our argument naturally constructs the best martingale that achieves $C_n(X_0)$. 
\begin{enumerate}
\item We know that the base case is $C_1(X) = 2X(1-X)$ (see~\figureref{transform-induct-bc} for this argument). 
\item Given the curve $C_{n-1}(X)$, we identify a geometric transformation $T$ (see~\figureref{transform-def}) that defines the curve $C_n(X)$ from the curve $C_{n-1}(X)$. 
  \sectionref{gap-main} summarizes the proof of this inductive step that crucially relies on the geometric interpretation of the problem, which is one of our primary technical contributions. 
  Furthermore, for any $n\geq 1$, there exist martingales such that its susceptibility is $C_n(X_0)$. 
\item Finally, \sectionref{gap-lower-upper-proof}
proves that the curve $C_n(X)$  lies above the curve $L_n(X) \defeq \frac2{\sqrt{2n-1}}X(1-X)$ and below the curve $U_n(X) \defeq \frac1{\sqrt n} \sqrt{X(1-X)}$. 
\end{enumerate}

\subsection{Proof of \theoremref{gap-main}}
\label{sec:gap-main} 

Our objective is the following. 
\begin{enumerate}
\item Given an arbitrary martingale $(X,E)$, find the maximum stopping time in this martingale, \ie, the stopping time $\tau_{\max}(X,E,1)$. 
\item For any depth $n$ and bias $X_0$, construct a martingale that achieves the max-score. 
  We refer to this martingale as the {\em optimal} martingale. 
  A priori, this martingale need not be unique. 
  However, we shall see that for each $X_0$, it is (essentially) a unique martingale. 
\end{enumerate} 
We emphasize that even if we are only interested in the exact value of $C_n(X_0)$ for $X_0=1/2$, it is unavoidable to characterize $C_{n-1}(X)$, for all values of $X\in[0,1]$. 
Because, in a martingale $(X_0=1/2,X_1,\dotsc,X_n)$, the value of $X_1$ can be arbitrary. 
So, without a precise characterization of the value $C_{n-1}(X_1)$, it is not evident how to calculate the value of $C_n(X_0=1/2)$. 
Furthermore, understanding $C_n(X_0)$, for all $X_0\in[0,1]$, yields entirely new applications for our result.

{\bfseries Base Case of $n=1$.} 
For a martingale $(X_0,X_1)$ of depth $n=1$, we have $X_1\in\zo$. 
Thus, without loss of generality, we assume that $E_1$ takes only two values (see~\figureref{transform-induct-bc}). 
Then, it is easy to verify that the max-score is always equal to $2X_0(1-X_0)$. 
This score is witnessed by the stopping time $\tau=1$. 
So, we conclude that 
  $\mathrm{opt}_1(X_0,1) = C_1(X_0) = 2X_0(1-X_0) $
\begin{figure}
\begin{center}
	\begin{tikzpicture}[auto,>=stealth, scale=1]
	\node[circle, minimum size=5mm, draw] (v) at (1,2) {$X_0$}; 
	\node[rectangle, minimum size=5mm, fill=gray!50!] (z) at (0,1) {$0$}; 
	\node[rectangle, minimum size=5mm, fill=gray!50!] (x) at (2,1) {$1$}; 
	\draw (v) to node[swap] {$1-X_0$} (z); 
	\draw (v) to node {$X_0$} (x);
\end{tikzpicture}
\caption{%
  Base Case for \theoremref{gap-main}. 
  Note $C_1(X_0) = \inf_{(X_0,X_1)} \sup_\tau \EX{ \abs{X_\tau-X_{\tau-1}} }$.
  The optimal stopping time is shaded and its score is $X_0\cdot\abs{1-X_0}+(1-X_0)\cdot\abs{0-X_0}$.%
  }
\label{fig:transform-induct-bc}
\end{center}
\end{figure}	


{\bfseries Inductive Step. $n=2$ (For Intuition).}
For simplicity, let us consider finite martingales, \ie, the sample space $\Omega_i$ of the random variable $E_i$ is finite. 
Suppose that the root $X_0=x$ in the corresponding martingale tree has $t$ children with values $x\p1, x\p2, \dotsc,x\p t$, and the probability of choosing the $j$-th child is $p\p j$, where $j\in\{1,\dotsc,t\}$ (see~\figureref{jump2-induct}). 
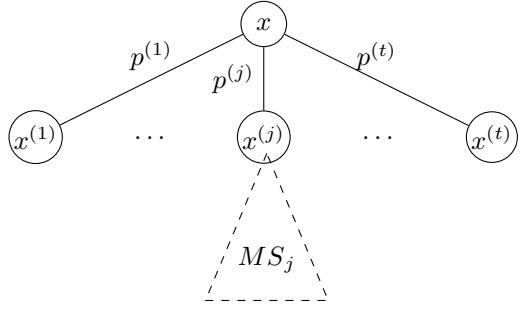
\begin{figure}
\begin{center}
\begin{tikzpicture}[sibling distance=1.5cm, level 2/.style={sibling distance =2cm}] 
\node [circle,draw,inner sep=0.5pt, minimum size=6mm]{$x$}
child{ node[circle,draw,inner sep=0.5pt, minimum size=6mm] {$x^{(1)}$} edge from parent node[above,draw=none] {$p^{(1)}$}
	}
child{ node[circle, draw = none] {$\ldots$} edge from parent[draw=none]
	{ node[draw = none] {} } 
}
child{ node[circle,draw,inner sep=0.5pt, minimum size=5mm] {$x^{(j)}$} 
	 edge from parent node[left,draw=none] {$p^{(j)}$}%
		{node[buffer]{${MS}_j $}  }%
}
child{ node[circle, draw = none] {$\ldots$}  edge from parent[draw=none]
	{ node[draw = none] {} } 
}
child{ node[circle,draw,inner sep=0.5pt, minimum size=6pt] {$x^{(t)}$}   edge from parent node[above,draw=none] {$p^{(t)}$}
}
;
\end{tikzpicture}
\caption{Inductive step for \theoremref{gap-main}. 
  ${MS}_j$ represents the max-score of the sub-tree of depth $n-1$ whose rooted at $x^{(j)}$. 
  For simplicity, the subtree of $x^{(j)}$ is only shown here.} 
\label{fig:jump2-induct}
\end{center}
\end{figure} 

Given a martingale $(X_0,X_1,X_2)$, the adversary's objective is to find the stopping time $\tau$ that maximizes the score $\EX{\abs{X_\tau - X_{\tau-1}}}$. 
If the adversary chooses to stop at $\tau=0$, then the score $\EX{\abs{X_\tau - X_{\tau-1}}}=0$, which is not a good strategy. 
So, for each $j$, the adversary chooses whether to stop at the child $x\p j$, or continue to a stopping time in the sub-tree rooted at $x\p j$. 
The adversary chooses the stopping time based on which of these two strategies yield a better score. 
If the adversary stops the martingale at child $j$, then the contribution of this decision to the score is $p\p j\abs{x\p j-x}$. 
On the other hand, if she does not stop at child $j$, then the contribution from the sub-tree is guaranteed to be $p\p jC_1(x\p j)$. 
Overall, from the $j$-th child, an adversary obtains a score that is at least $p\p j\max\left\{\abs{x\p j-x},C_1(x\p j)\right\}$. 
\begin{figure}
\begin{center}
\begin{tikzpicture}[domain=0:1,auto,>=stealth,scale=7.5] 
\coordinate (O) at (0,0); 
\draw[->,name path=xaxis] (0,0) to (1.1,0) node[below] {$X$-axis}; 
\draw[->,name path=yaxis] (0,0) to (0,.7) node[above] {$Y$-axis} coordinate (ytop); 
\draw[->,name path=ypaxis,draw=none] (1,0) to (1,.75) coordinate (yptop);

\coordinate (X) at (.3,0); 
\node at (X) {\textbullet}; 
\node[anchor=north west] (x)  at ($(X) +(-0.05,0)$) {$X=(x,0)$};  
\draw[dashed, name path=xvert] (X) to ++(90:0.6) coordinate (top); 

\draw[dotted,name path=curve,thick] plot (\x,{2*\x*(1-\x)}) node[anchor=south west] {$C_1$}; 
\draw [name intersections={of=curve and xvert, by=m}] coordinate (mid) at (m);

\draw[dashed,->,name path=ell0,draw=none] (X) to ++(135:0.5);
\draw[dashed,->,name path=ell1,draw=none] (X) to ++(45:1.1);

\draw [name intersections={of=curve and ell0, by=P1}] node at (P1) {\textbullet};
\node[anchor=east] at (P1) {$P_1$};
\draw [name intersections={of=yaxis and ell0, by=L}] node at (L) {\textbullet};
\node[anchor=east] at (L) {$L$};

\draw[dotted,thick] (X) to (L);
\draw[thick] (P1) to (L); 

\draw [name intersections={of=curve and ell1, by=P2}] node at (P2) {\textbullet};
\node[anchor=west] at (P2) {$P_2$};
\draw [name intersections={of=ypaxis and ell1, by=R}] node at (R) {\textbullet};
\node[anchor=west] at (R) {$R$};

\draw[dotted,thick] (X) to (R);
\draw[thick] (P2) to (R); 

\draw[dashed,name path=p1p2] (P1) to (P2); 
\draw [name intersections={of=p1p2 and xvert, by=Q}] node at (Q) {\textbullet};
\node [anchor=north west] at (Q) {$Q$};

\begin{pgfonlayer}{background}
\begin{scope}
\clip (P1) rectangle (mid); 
\draw[thick,fill=white] plot (\x,{2*\x*(1-\x)});
\end{scope}

\begin{scope}
\clip (P2) rectangle (top); 
\draw[thick,fill=white] plot (\x,{2*\x*(1-\x)});
\end{scope} 
\end{pgfonlayer}

\draw (X) ++(-0.05,0) arc [start angle=180, end angle=135, radius=0.05];
\draw (X) ++(180-20:0.1) node {$\pi/4$};
\draw (X) ++(0.05,0) arc [start angle=0, end angle=45, radius=0.05]; %
\draw (X) ++(20:0.1) node {$\pi/4$};

\foreach \x/\y/\i/\where in {0.075/0.225/1/south,
                             0.5/0.5/2/south, 
                             0.8/0.5/3/west,
                             0.9/0.6/4/west}{
  \node (x\i) at (\x,\y) {$\circ$}; \node[anchor=\where] at (x\i) {$Z^{(\i)}$};
}
\draw[dashed,pattern=dots, pattern color=gray] (x1.center) to (x2.center) to (x4.center) to (x3.center) to cycle; 
\node (x01) at (0.075,0) {$\circ$}; \node[anchor=south] at (x01) {$x^{(1)}$}; 
\node (x02) at (.5,0) {$\circ$}; \node[anchor=south west] at (x02) {$x^{(2)}$};
\node (x03) at (.8,0) {$\circ$}; \node[anchor=south] at (x03) {$x^{(3)}$};
\node (x04) at (.9,0) {$\circ$}; \node[anchor=south] at (x04) {$x^{(4)}$};
\draw[name path=temp,draw=none] (x1.center) to (x3.center);
\draw[name path=ntemp,draw=none] (x1.center) to (x2.center); 
\draw [name intersections={of=temp and xvert, by=Qp}] node at (Qp) {$\otimes$} node [anchor=west] at (Qp) {$Q'$};
\draw [name intersections={of=ntemp and xvert, by=Qp}] node at (Qp) {$\otimes$} node [anchor= south east] at (Qp) {$Q''$};

\end{tikzpicture}
\end{center}
\caption{Intuitive summary of the inductive step for $n=2$.}
\label{fig:induct-c1c2}
\end{figure}
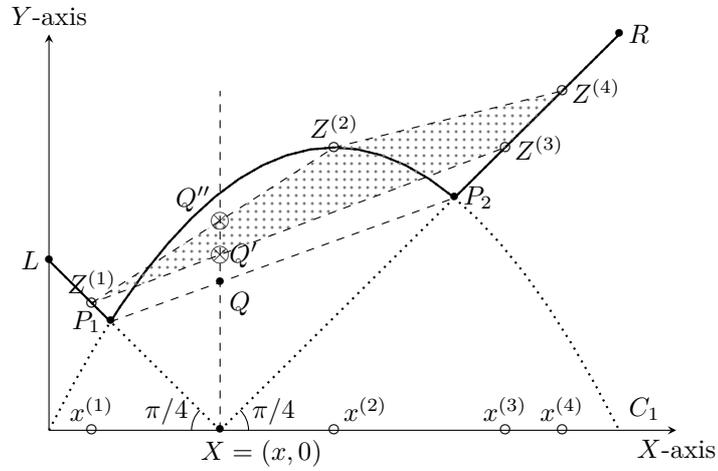  

Let $h\p j\defeq \max\left\{\abs{x\p j-x},C_1(x\p j)\right\}$. 
We represent the points $Z\p j=(x\p j,h\p j)$ in a two dimensional plane. 
Then, clearly all these points lie on the solid curve defined by $\max\left\{\abs{X-x},C_1(X)\right\}$, see~\figureref{induct-c1c2}.

Since $(X,E)$ is a martingale, we have $x=\sum_{j=1}^t p\p jx\p j$ and the adversary's strategy for finding $\tau_{\max}$ gives us $\mathrm{max\textnormal{-}score}_1(X,E)= \sum_{j=1}^tp\p jh\p j$. 
This observation implies that the coordinate $(x,\mathrm{max\textnormal{-}score}_1(X,E)) = \sum_{j=1}^tp\p jZ\p j$. 
So, the point in the plane giving the adversary the maximum score for a tree of depth $n=2$ with bias $X_0=x$ lies in the {\em intersection} of the convex hull of the points $Z\p1,\dotsc,Z\p t$, and the line $X=x$. 
Let us consider the martingale defined in \figureref{induct-c1c2} as a concrete example. 
Here $t=4$, and the points $Z\p1, Z\p2, Z\p3, Z\p4$ lie on $\max\left\{\abs{X-x},C_1(X)\right\}$.
The martingale designer specifies the probabilities $p^{(1)}, p^{(2)}, p^{(3)}$, and $p^{(4)}$, such that $p^{(1)}x^{(1)} + \dotsi + p^{(4)}x^{(4)}=x$. 
These probabilities are not represented in \figureref{induct-c1c2}. 
Note that the point $\left( p^{(1)}x^{(1)} + \dots + p^{(4)}x^{(4)}, p^{(1)}h^{(1)} +\dots + p^{(4)}h^{(4)} \right)$ representing the score of the adversary is the point $p^{(1)}Z^{(1)}+\dots+p^{(4)}Z^{(4)}$. 
This point lies inside the convex hull of the points $Z^{(1)}, \dots ,Z^{(4)}$ and on the line $X = p^{(1)}x^{(1)} + \dots + p^{(4)}x^{(4)}=x$. 
The exact location depends on $p^{(1)},\dots,p^{(4)}$.

The point $Q'$ is the point with minimum height. 
Observe that the height of the point $Q'$ is at least the height of the point $Q$. 
So, in any martingale, the adversary shall find a stopping time that scores more than (the height of) the point $Q$. 

On the other hand, the martingale designer's objective is to reduce the score that an adversary can achieve. 
So, the martingale designer chooses $t=2$, and the two points $Z\p1=P_1$ and $Z\p2=P_2$ to construct the optimum martingale. 
We apply this method for each $x\in[0,1]$ to find the corresponding point $Q$. 
That is, the {\em locus of the point} $Q$, for $x\in[0,1]$, yields the curve $C_2(X)$. 

We claim that the height of the point $Q$ is the {\em harmonic-mean} of the heights of the points $P_1$ and $P_2$. 
This claim follows from elementary geometric facts. 
Let $h_1$ represent the height of the point $P_1$, and $h_2$ represent the height of the point $P_2$. 
Observe that the distance of $x - x_S(x) = h_1$ (because the line $\ell_1$ has slope $\pi - \pi/4$). 
Similarly, the distance of $x_L(x)-x=h_2$ (because the line $\ell_2$ has slope $\pi/4$). 
So, using properties of similar triangles, the height of $Q$ turns out to be
  $$ h_1 + \frac{h_1}{h_1+h_2}\cdot(h_2-h_1) = \frac{2h_1h_2}{h_1+h_2}.$$

This property inspires the definition of the geometric transformation $T$, see~\figureref{transform-def}.
Applying $T$ on the curve $C_1(X)$ yields the curve $C_2(X)$ for which we have $C_2(x) = \mathrm{opt}_2(x,1)$. 
\begin{figure}%
\begin{boxedminipage}{\linewidth}
\input{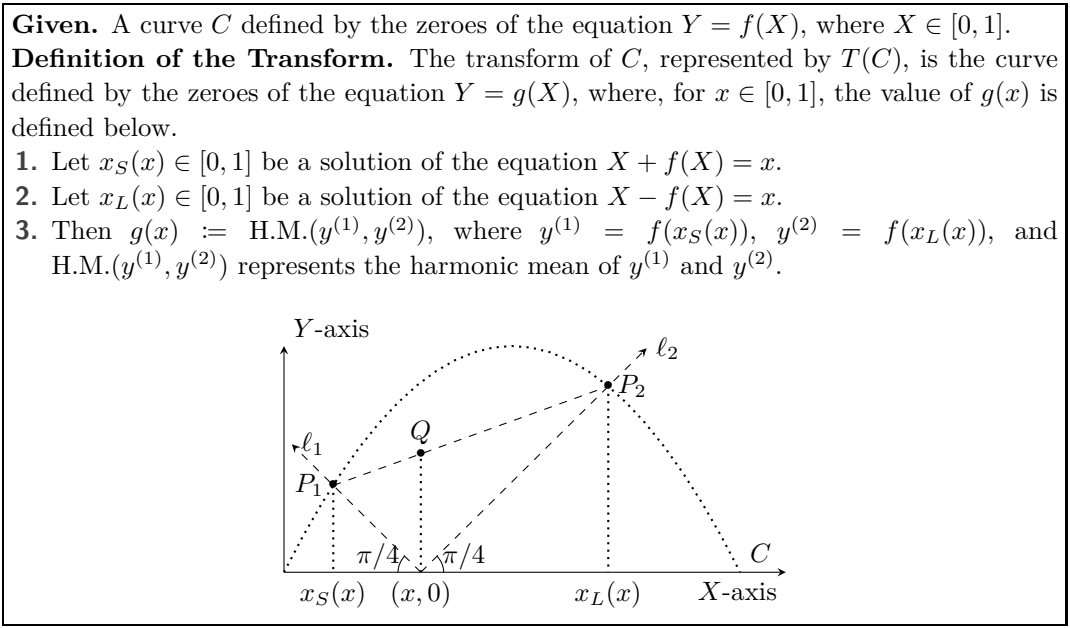} 
\begin{center}
\begin{tikzpicture}[domain=0:1,scale=6,auto,>=stealth] 
\coordinate (O) at (0,0); 
\draw[->,name path=xaxis] (0,0) to (1.1,0) node[anchor=north east] {$X$-axis}; 
\draw[->,name path=yaxis] (0,0) to (0,.5) node[anchor=south west] {$Y$-axis}; 

\draw[dotted,name path=curve,thick] plot (\x,{2*\x*(1-\x)}) node[right,anchor=south west] {$C$}; 

\coordinate (X) at (.3,0); 

\node[anchor=north] (x) at (X) {$(x,0)$};  

\draw[dashed,->,name path=ell0] (X) to ++(135:0.4) node [right] {$\ell_1$};
\draw[dashed,->,name path=ell1] (X) to ++(45:0.7) node [right] {$\ell_2$}; 

\draw [name intersections={of=curve and ell0, by=P0}] node at (P0) {\textbullet};
\node[anchor=east] at (P0) {$P_1$};

\draw [name intersections={of=curve and ell1, by=P1}] node at (P1) {\textbullet};
\node[anchor=west] at (P1) {$P_2$};

\draw[dashed,name path=p0p1] (P0) to (P1); 
\draw[name path=vert,draw=none] (X) to +(90:0.5); 

\draw [name intersections={of=p0p1 and vert, by=Q}] node at (Q) {\textbullet};
\node [anchor=south] at (Q) {$Q$}; 
\draw[dotted,thick] (X) to (Q); 

\draw (X) ++(-0.05,0) arc [start angle=180, end angle=135, radius=0.05];
\draw (X) ++(180-20:0.1) node {$\pi/4$};
\draw (X) ++(0.05,0) arc [start angle=0, end angle=45, radius=0.05]; %
\draw (X) ++(20:0.1) node {$\pi/4$};

\draw[dotted,thick] (P0) to (P0 |- O); 
\node[anchor=north] at (P0 |- O) {$x_S(x)$};
\draw[dotted,thick] (P1) to (P1 |- O); 
\node[anchor=north] at (P1 |- O) {$x_L(x)$};
\end{tikzpicture}
\end{center}
\end{boxedminipage}
\caption{Definition of transform of a curve $C$, represented by $T(C)$. 
  The locus of the point $Q$ (in the right figure) defines the curve $T(C)$.}
\label{fig:transform-def}
\end{figure}

{\bfseries General Inductive Step.} 
Note that a similar approach works for general $n=d\geq 2$. 
Fix $X_0$ and $n=d\geq 2$. 
We assume that the adversary can compute $C_{d-1}(X_1)$, for any $X_1\in[0,1]$. 

Suppose the root in the corresponding martingale tree has $t$ children with values $x\p1, x\p2,\dotsc,x\p t$, and the probability of choosing the $j$-th child is $p\p j$ (see~\figureref{jump2-induct}).
Let $(X\p j,E\p j)$ represent the martingale associated with the sub-tree rooted at $x\p j$. 

\begin{figure}
\begin{center}
\begin{tikzpicture}[domain=0:1,auto,>=stealth,scale=7.5] 
\coordinate (O) at (0,0); 
\draw[->,name path=xaxis] (0,0) to (1.1,0) node[below] {$X$-axis}; 
\draw[->,name path=yaxis] (0,0) to (0,.7) node[above] {$Y$-axis} coordinate (ytop); 
\draw[->,name path=ypaxis,draw=none] (1,0) to (1,.75) coordinate (yptop);

\coordinate (X) at (.3,0); 
\node at (X) {\textbullet}; 
\node[anchor=north west] (x)  at ($(X) +(-0.05,0)$) {$X=(x,0)$};  
\draw[dashed, name path=xvert] (X) to ++(90:0.6) coordinate (top); 

\draw[dotted,name path=curve,thick] plot (\x,{2*\x*(1-\x)}) node[anchor=south west] {$C_d$}; 
\draw [name intersections={of=curve and xvert, by=m}] coordinate (mid) at (m);

\draw[dashed,->,name path=ell0,draw=none] (X) to ++(135:0.5);
\draw[dashed,->,name path=ell1,draw=none] (X) to ++(45:1.1);

\draw [name intersections={of=curve and ell0, by=P0}] node at (P0) {\textbullet};
\node[anchor=east] at (P0) {$P_1$};
\draw [name intersections={of=yaxis and ell0, by=L}] node at (L) {\textbullet};
\node[anchor=east] at (L) {$L$};

\draw[dotted,thick] (X) to (L);
\draw[thick] (P0) to (L); 

\draw [name intersections={of=curve and ell1, by=P1}] node at (P1) {\textbullet};
\node[anchor=west] at (P1) {$P_2$};
\draw [name intersections={of=ypaxis and ell1, by=R}] node at (R) {\textbullet};
\node[anchor=west] at (R) {$R$};

\draw[dotted,thick] (X) to (R);
\draw[thick] (P1) to (R); 

\draw[dashed,name path=p0p1] (P0) to (P1); 
\draw [name intersections={of=p0p1 and xvert, by=Q}] node at (Q) {\textbullet};
\node [anchor=north west] at (Q) {$Q$};

\begin{pgfonlayer}{background}
\draw [fill=gray!50!,draw=none] (ytop) to (L) to (P0) to (mid) to (P1) to (R) to (yptop) to cycle; 

\begin{scope}
\clip (P0) rectangle (mid); 
\draw[thick,fill=white] plot (\x,{2*\x*(1-\x)});
\end{scope}

\begin{scope}
\clip (P1) rectangle (top); 
\draw[thick,fill=white] plot (\x,{2*\x*(1-\x)});
\end{scope} 
\end{pgfonlayer}

\draw (X) ++(-0.05,0) arc [start angle=180, end angle=135, radius=0.05];
\draw (X) ++(180-20:0.1) node {$\pi/4$};
\draw (X) ++(0.05,0) arc [start angle=0, end angle=45, radius=0.05]; %
\draw (X) ++(20:0.1) node {$\pi/4$};

\foreach \x/\y/\i/\where in {0.1/0.3/1/south east, 
                             0.6/0.5/2/south west, 
                             0.7/0.6/3/west, 
                             0.5/0.65/4/south, 
                             0.15/0.5/5/east, 
                             0.5/.55/6/east,
                             0.2/0.45/7/west}{
  \node (x\i) at (\x,\y) {$\circ$}; \node[anchor=\where] at (x\i) {$Z^{(\i)}$};
}
\draw[dashed,pattern=dots, pattern color=gray] (x1.center) to (x2.center) to (x3.center) to (x4.center) to (x5.center) to cycle; 

\draw[name path=temp,draw=none] (x1.center) to (x2.center); 
\draw [name intersections={of=temp and xvert, by=Qp}] node at (Qp) {$\otimes$} node [anchor=north west] at (Qp) {$Q'$};

\end{tikzpicture}
\end{center}
\caption{Intuitive Summary of the inductive argument. 
  Our objective is to pick the set of points $\{Z\p 1, Z\p2\dotsc\}$ in the gray region to minimize the length of the intercept $XQ'$ of their (lower) convex hull on the line $X=x$. 
  Clearly, the unique optimal solution corresponds to including both $P_1$ and $P_2$ in the set.%
}
\label{fig:induct}
\end{figure}
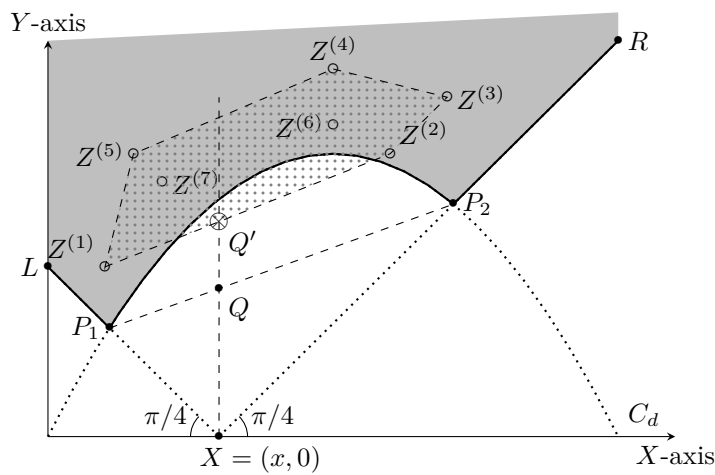

For any $j\in\{1,\dotsc,t\}$, the adversary can choose to stop at the child $j$. 
This decision will contribute $\abs{x\p j-x}$ to the score with weight $p\p j$. 
On the other hand, if she continues to the subtree rooted at $x\p j$, she will get at least a contribution of $\mathrm{max\text{-}score}_1(X\p j,E\p j)$ with weight $p\p j$. 
Therefore, the adversary can obtain the following contribution to her score
  $$p\p j\max\left\{\abs{x\p j-x},C_{d-1}(x\p j)\right\}$$

Similar to the case of $n=2$, we define the points $Z\p1,\dotsc,Z\p t$. 
For $n> 2$, however, there is one difference from the $n=2$ case. 
The point $Z\p j$ need not {\em lie on the solid curve}, but it can lie on or above it, \ie, they lie in the gray area of \figureref{induct}. 
This phenomenon is attributable to a suboptimal martingale designer producing martingales with suboptimal scores, \ie, {\em strictly above} the solid curve. 
For $n=1$, it happens to be the case that, there is (effectively) only one martingale that the martingale designer can design (the optimal tree).  
The adversary obtains a score that is at least the height of the point $Q'$, which is at least the height of $Q$.
On the other hand, the martingale designer can choose $t=2$, and $Z\p 1=P_1$ and $Z\p2=P_2$ to define the optimum martingale. 
Again, the locus of the point $Q$ is defined by the curve $T(C_{d-1})$.

{\bfseries Conclusion.} 
So, by induction, we have proved that $C_n(X) = T^{n-1}(C_1(X))$. 
Additionally, note that, during induction, in the optimum martingale, we always have $\abs{x\p0-x}=C_{n-1}(x\p0)$ and $\abs{x\p1-x}=C_{n-1}(x\p1)$. 
Intuitively, the decision to stop at $x\p j$ or continue to the subtree rooted at $x\p j$ has identical consequence. 
So, by induction, {\em all stopping times} in the optimum martingale have score $C_n(x)$. 

\appendixref{geometric-proof-l1} provides a more technical proof. 


\subsection{Estimation of $C_n(X)$ : Proof of \lemmaref{gap-lower-upper}}
\label{sec:gap-lower-upper-proof} 
In this section, we prove \lemmaref{gap-lower-upper}, which tightly estimates the curve $C_n$.

Recall that we defined $L_n(X) = \frac2{\sqrt{2n-1}} X(1-X)$ and $U_n(X) = \frac1{\sqrt n}\sqrt{X(1-X)}$. 
Our objective is to inductively prove that $U_n \succcurlyeq C_n \succcurlyeq L_n$. 
To this end, we define the curve $G_n\defeq a_nX(1-X)$ where $a_1=2$ and $a_{n+1}=2\left(\frac{\sqrt{a_n^2+1}-1}{a_n}\right)$. Notice that $G_1(X)=L_1(X)$ for all $X\in [0,1]$. Moreover, it follows from \lemmaref{lowerbound-an} that $a_n\geq \frac{2}{\sqrt{2n-1}}$, and so $G_n \succcurlyeq L_n$. Observe that since we do not have a closed form for $G_n$, we use $L_n$ as a lower bound.

\begin{proof}
Since $G_n \succcurlyeq L_n$, it is sufficient to prove by induction that $U_n \succcurlyeq C_n \succcurlyeq G_n$.

{\bfseries Base Case of $n=1$.} 
Since, $C_1(X)=G_1(X) = 2X(1-X)$, it is obvious that $C_1\succcurlyeq G_1$. 
Moreover, we know that $U_1(X) = \sqrt{X(1-X)}$. 
It is easy to verify that $U_1(X) \geq C_1(X)$ for all $X\in [0,1]$ which is equivalent to $U_1 \succcurlyeq C_1$.

{\bfseries Inductive Argument.}
 
Suppose we have $U_n\succcurlyeq C_n\succcurlyeq G_n$. 
Then, we have $T(U_n)\succcurlyeq T(C_n)\succcurlyeq T(G_n)$ (by \claimref{above-to-above}).
Note that $C_{n+1}=T(C_n)$. 
We shall prove that $T(G_n)\succcurlyeq G_{n+1}$, and $U_{n+1}\succcurlyeq T(U_n)$ (refer to \claimref{our-induct} and \claimref{our-induct-upper}) respectively. 
Consequently, it follows that $U_{n+1}\succcurlyeq C_{n+1}\succcurlyeq G_{n+1}$. 
\figureref{summary-curves} pictorially summarizes this argument.

\end{proof}

\begin{figure}[H]
\begin{center}
\begin{tikzpicture}[>=stealth, auto, yscale=0.8, xscale=1.2]\footnotesize
\node (ci) at (0,0) {$C_i$}; 
\node (li) at (3,0) {$G_i$}; 
  \node at ($(ci)!0.5!(li)$) {$\succcurlyeq$}; 
\node (Tli) at (3,-2) {$T(G_i)$}; 
  \draw[->] (li) to node {$T$} (Tli); 
\node (li+1) at (3,-4) {$G_{i+1}$}; 
  \node (a) at ($(Tli)!0.5!(li+1)$) {\rotatebox[origin=c]{-90}{$\succcurlyeq$}};
  \node[anchor=west] at (a.east) {\claimref{our-induct}}; 
\node (Tci) at (0,-2) {$T(C_i)$}; 
  \draw[->] (ci) to node {$T$} (Tci);
\node (ci+1) at (0,-4) {$C_{i+1}$};
  \node at ($(Tci)!0.5!(ci+1)$) {\rotatebox[origin=c]{-90}{$=$}};
\node (ui) at (-3,0) {$U_i$}; 
  \node at ($(ui)!0.5!(ci)$) {$\succcurlyeq$}; 
\node (Tui) at (-3,-2) {$T(U_i)$}; 
  \draw[->] (ui) to node {$T$} (Tui); 
\node (ui+1) at (-3,-4) {$U_{i+1}$}; 
  \node (b) at ($(Tui)!0.5!(ui+1)$) {\rotatebox[origin=c]{90}{$\succcurlyeq$}};
  \node[anchor=east] at (b.west) {\claimref{our-induct-upper}};

  \node (c) at ($(Tui)!0.5!(Tci)$) {$\succcurlyeq$};
    \node[anchor=south] at (c.north) {\claimref{above-to-above}}; 
  \node (d) at ($(Tci)!0.5!(Tli)$) {$\succcurlyeq$};
    \node[anchor=south] at (d.north) {\claimref{above-to-above}}; 

\draw[->,dashed,rounded corners=20pt] (ci+1.north east) to (Tci.south east) to (Tli.south west) to (li+1.north west); 
\draw[->,dashed,rounded corners=20pt] (ui+1.north east) to (Tui.south east) to (Tci.south west) to (ci+1.north west); 
\end{tikzpicture}
\end{center}
\caption{The outline of the inductive proof demonstrating that if the curves $U_i$ and $G_i$ sandwich the curve $C_i$, then the curves $U_{i+1}$ and $G_{i+1}$ sandwich the curve $C_{i+1}$.
Recall that the notation ``$A \succcurlyeq B$'' implies that the curve $A$ lies on-or-above the curve $B$.%
}
\label{fig:summary-curves} 
\end{figure}
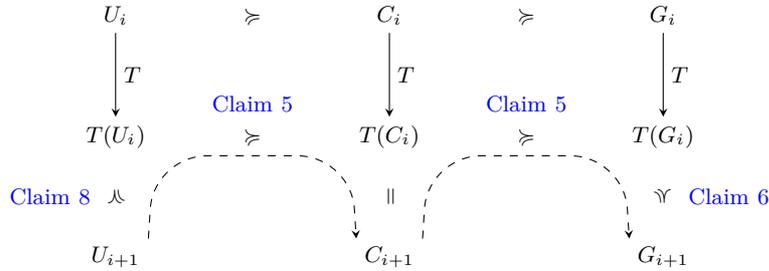

\begin{claim}
	\label{clm:above-to-above} 
	Let $C$ and $D$ be concave downward curves in the domain $X\in[0,1]$, and both curves $C$ and $D$ are above the axis $Y=0$ and contain the points $(0,0)$ and $(1,0)$. 
	Let $C$ and $D$ be curves such that $C\succcurlyeq D$ in the domain $X\in[0,1]$, then the curve $T(C)\succcurlyeq T(D)$. 
\end{claim}

\begin{proof}
	See \figureref{above-to-above-pic}.\\
	Observe that if the curves $C$ and $D$ are identical, then the result holds. 
	So, let us assume that $C$ and $D$ are not identical. 
	Note that if we have two distinct concave curves $C$ and $D$ such that $C \succcurlyeq D$ then these two curves cannot intersect at any additional point in the domain $(0,1)$. 
	Fix $x\in(0,1)$. 
	Let $Q_C=(x,y_C)$ be the intersection of the curve $T(C)$ with the line $X=x$. 
	Similarly, let $y_D$ be the intersection of the curve $T(D)$ with the line $X=x$. 
	Let $P$ be the point $(x,0)$. 
	Let $\ell_0$ be the ray starting at $P$ with slope $135$-degrees. 
	Let $\ell_1$ be the ray starting at $P$ with slope $45$-degrees. 
	Let $\ell_0$ intersect the curves $D$ and $C$ at $L_{D}$ and $L_{C}$, respectively. 
	And, let $\ell_1$ intersect the curves $D$ and $C$ at $R_D$ and $R_C$, respectively. 
	Observe in the triangles $\Delta PL_CR_C$ and $\Delta PL_DR_D$ the line segment $L_CR_C$ does not intersect with the line segment $L_DR_D$. 
	Otherwise, if the line segments $L_CR_C$ intersects with $L_DR_D$, then the distinct concave curves $C$ and $D$ intersect at some point with X-coordinate in $(0,1)$ as well (a contradiction). \newline
	Therefore, we have $L_CR_C \succcurlyeq L_DR_D$. 
	Note that $y_C$ is the intersection of $L_CR_C$ with $X=x$, and $y_D$ is the intersection of $L_DR_D$ with $X=x$. 
	So, we have $y_C \succcurlyeq y_D$. 
\end{proof}

\begin{figure}[H]
  
\begin{center}
\begin{tikzpicture}[>=stealth,auto,scale=5,domain=0:1] \footnotesize
  \coordinate (O) at (0,0); 
  \coordinate (X) at (0.3,0); 
  
  \draw [name path=vert,dashed] (X) to ++(0,0.6); 
  
  \draw[->] (O) to ++(1.1,0); 
  \draw[->] (O) to ++(0,.6);
  
  \draw[name path=C] plot (\x,{2*\x*(1-\x)}) node[right,anchor=south west] {$C$};
  
  \draw[name path=D] plot (\x,{\x*(1-\x)}) node[right,anchor=south east] {$D~~~$};
  
  \draw[name path=ell0,->] (X) to ++(135:0.4) node[anchor=south west] {$\ell_0$}; 
  
  \draw[name path=ell1,->] (X) to ++(45:0.7) node[anchor=east] {$\ell_1$}; 
  
  \draw (X) ++(-0.05,0) arc [start angle=180, end angle=135, radius=0.05];
  \draw (X) ++(180-20:0.1) node {$\pi/4$};
  \draw (X) ++(0.05,0) arc [start angle=0, end angle=45, radius=0.05]; %
  \draw (X) ++(20:0.1) node {$\pi/4$};
  
  \draw [name intersections={of=C and ell0, by=LC}] node at (LC) {\textbullet} node at (LC) [anchor=east] {$L_C$};
  \draw [name intersections={of=D and ell0, by=LD}] node at (LD) {\textbullet} node at (LD) [anchor=north] {$L_D$};
  
  \draw [name intersections={of=C and ell1, by=RC}] node at (RC) {\textbullet} node at (RC) [anchor=west] {$R_C$};
  \draw [name intersections={of=D and ell1, by=RD}] node at (RD) {\textbullet} node at (RD) [anchor=north west] {$R_D$};
  
  \draw[name path=LCRC,dashed] (LC) to (RC);
  \draw[name path=LDRD,dashed] (LD) to (RD); 
  
  \draw [name intersections={of=LCRC and vert, by=QC}] node at (QC) {\textbullet} node at (QC) [anchor=south west] {$Q_C$};
  \draw [name intersections={of=LDRD and vert, by=QD}] node at (QD) {\textbullet} node at (QD) [anchor=north west] {$Q_D$};

\end{tikzpicture}
\end{center}
\caption{%
    Summary of the proof of \claimref{above-to-above}.%
    }
\label{fig:above-to-above-pic}
\end{figure}

In the following claim, we show that the transformation of a curve whose characteristics are specified below, will be ``above'' the curve itself. 

\begin{claim}
	\label{clm:our-induct} 
	Let $F_n$ be the curve above $Y=0$ defined by the zeros of the equation $Y = f_n X(1-X)$,
	where $f_1>0$ and $f_{n+1}=2\left(\frac{\sqrt{f_n^2+1}-1}{f_n}\right)$ for all $n\geq 1$.
	Then, we have $T(F_n) \succcurlyeq F_{n+1}$. 
\end{claim}

\begin{proof}
	For each $k$, the curve $F_k$ is a concave downward curve that contains the points $(0,0)$ and $(1,0)$, so based on \claimref{preserve-concavity}, for each $k$, the curve $T(F_k)$ is also concave downward and contains the points $(0,0)$ and $(1,0)$.

	Let us fix $x\in[0,1]$ and let $x_0\in [0,1]$ denotes the smaller root of the two roots of the equation $x^* + f_n x^*(1-x^*) = x$ and let $y_0$ be the value $f_nx_0(1-x_0)$.
	Moreover, let $x_1\in [0,1]$ denotes the larger root of the two roots of the equation $x^* - f_n x^*(1-x^*) = x$ and let $y_1$ be the value $f_nx_1(1-x_1)$. 
	So, we have 
	$$x_0 = \frac{(f_n+1) - \sqrt{(f_n+1)^2 -4xf_n}}{2f_n} \; ,$$
	and
	\begin{align*}
		y_0 &= \frac{\left(f_n+1 - \sqrt{(f_n+1)^2 -4xf_n}\right)\left(f_n-1 + \sqrt{(f_n+1)^2 -4xf_n}\right)}{4f_n}\\
		&= \frac{f_n^2 - 1 - (f_n+1)^2 + 4xf_n + 2\sqrt{(f_n+1)^2 -4xf_n}}{4f_n}\\
		&= \frac{(2x-1)f_n -1  + \sqrt{(f_n+1)^2 -4xf_n}}{2f_n}
	\end{align*}
	and since $F_n(x)=F_n(1-x)$ (i.e. $F_n$ is a symmetric curve around $\frac12$), $y_1$ can be found by replacing $x$ with $1-x$ in the 
	formula that we found for $y_0$. 
	\begin{align*}
		y_1 &= \frac{(2(1-x)-1)f_n -1  + \sqrt{(f_n+1)^2 -4(1-x)f_n}}{2f_n}\\
		&= \frac{(1-2x)f_n -1  + \sqrt{(f_n-1)^2 +4xf_n}}{2f_n}
	\end{align*}
	
	To prove the claim, it suffices to show that the harmonic mean of $y_0$ and $y_1$ is at least equal to $f_{n+1}\cdot x(1-x) $.
	We make the substitution $x=1/2 - z$ and we need to consider only $z\in[0,1/2]$ because as mentioned earlier the curves are symmetric around the line $X=1/2$. 
	From this substitution, we get
	\begin{align*}
		y_0 &= \frac{-2zf_n -1 + \sqrt{f_n^2+1 +4zf_n}}{2f_n}\\
		&= \frac{\left( f_n^2+1 +4zf_n\right) - (1+2zf_n)^2}{2f_n\left(\sqrt{f_n^2+1 +4zf_n} + (1+2zf_n)\right)}\\
		&= \frac{(1-4z^2)}{2f_n}\cdot\frac{f_n^2}{\sqrt{f_n^2+1 +4zf_n} + (1+2zf_n)}
		\end{align*}
		And also, 
		\begin{align*}
		y_1 &= \frac{2zf_n -1 + \sqrt{f_n^2+1 -4zf_n}}{2f_n}\\
		&= \frac{(1-4z^2)}{2f_n}\cdot\frac{f_n^2}{\sqrt{f_n^2+1 -4zf_n} + (1-2zf_n)}\\
	\end{align*}
	Let us define, 
	\begin{align*}
		\ell &\defeq f_{n+1}x(1-x) 
	\end{align*}
	Notice that $\ell = \frac{(1-4z^2)}{2f_n}\left(\sqrt{f_n^2+1} -1 \right)$.
	So, we need to prove the following
	\begin{align*}
		\mathrm{H.M.}\left(y_0,y_1\right) &\geq \ell\\
		\mathrm{H.M.}\left(\frac{2f_n}{1-4z^2}\cdot y_0,\frac{2f_n}{1-4z^2}\cdot y_1\right) &\geq \frac{2f_n}{1-4z^2}\cdot \ell\\
		\mathrm{H.M.}\left(\frac{f_n^2}{\sqrt{f_n^2+1 +4zf_n} + (1+2zf_n)},\frac{f_n^2}{\sqrt{f_n^2+1 -4zf_n} + (1-2zf_n)}\right) &\geq \sqrt{f_n^2+1}-1\\
		\mathrm{A.M.}\left(\frac{\sqrt{f_n^2+1 +4zf_n} + (1+2zf_n)}{f_n^2},\frac{\sqrt{f_n^2+1 -4zf_n} + (1-2zf_n)}{f_n^2}\right) &\leq \frac1{\sqrt{f_n^2+1}-1}\\
		\mathrm{A.M.}\left({\sqrt{f_n^2+1 +4zf_n} },{\sqrt{f_n^2+1 -4zf_n} }\right) +1&\leq {\sqrt{f_n^2+1}+1} \; .
	\end{align*}
	We used the fact that for any $a,b$, $\mathrm{H.M.}(a,b)=\frac{2ab}{a+b}
	=\frac1{\mathrm{A.M.}(a,b)}$.
	The last inequality is due to RMS-AM inequality. 
\end{proof}

In the following claim, we show that the geometric transformation $T$ preserves some characteristics of the curve it is transforming - specifically if the original curve was concave downward and symmetric around $\frac12$ then the new curve obtained will also retain these properties.
\begin{claim}
	\label{clm:preserve-concavity}
	Suppose the curve $C$ which is concave downward in the interval $X\in [0,1]$ and symmetric around $\frac12$, and the points $(0,0)$ and $(1,0)$ lie on it - is given. Suppose the curve $F$ is a curve defined by applying transformation $T$, defined in \figureref{transform-def}, on curve $C$. Then, $F$ has the same properties i.e. $F$ is also concave downward, symmetric around $\frac12$, and contains the points $(0,0)$ and $(1,0)$.   
\end{claim}
\begin{proof}
	Since $C$ is symmetric around $\frac12$, the point $(x,y)$ lies on the curve if and only if the point $(1-x,y)$ lies on the curve. Suppose that the curve $C$ is defined by the zeros of the equation $Y=f(X)$ and 
	the curve $T(C)$ is defined by the zeros of the equation $Y=g(X)$. Then, according to the definition of the transformation $T$, $g(x)=\mathrm{H.M.}(y^{(1)},y^{(2)})$ where $y^{(1)}=f(x^{(1)})$ and $y^{(2)}=f(x^{(2)})$ where $x^{(1)}$ is the solution of $X+f(X)=x$ and  $x^{(2)}$ is the solution of $X-f(X)=x$. Note that since $f(x)=f(1-x)$, we have that $1-x^{(1)}$ is the solution of $X-f(X)=1-x$ and $1-x^{(2)}$ is the solution of the equation $X+f(X)=1-x$. Similarly, since $f(1-x^{(1)})=f(x^{(1)})=y^{(1)}$ and $f(1-x^{(2)})=f(x^{(2)})=y^{(2)}$, it follows that $g(1-x)=\mathrm{H.M.}(y^{(2)},y^{(1)})=g(x)$ which implies that $T(C)$ is symmetric around $\frac12$. For $x=0$ or $x=1$, $y^{(1)}$ or $y^{(2)}$ is $0$ and so $g(1)=g(0)=0$.

	We provide a geometric proof to show that $F$ is concave downwards and use \figureref{concavity-preserved1} as illustration. We know that a curve is concave downward in an interval if and only if the line that joins any two points of the curve is below the curve. Let us fix $x_1\leq x_3$ in $[0,1]$, see \figureref{concavity-preserved1}. The height of points $H_1,H_2,H_3$ are respectively the value of $T(C)$ at points $x_1,x_2,x_3$ respectively. Our goal is to show that $H_2$ is above the segment $H_1H_3$ for any choice of $x_2$. Observe that,
	$H_2$ lies on the segment $P_2Q_2$. Since $C$ is concave down, the segment $P_2Q_2$ is above the segment $F_1F_2$. Note that we are fixing $x_1$ and $x_3$ and allowing $x_2$ to change between $x_1$ and $x_3$. Then, we see that the segment $F_1F_2$ changes from $P_1Q_1$ to $P_3Q_3$ and is always above the segment $H_1H_3$.
\end{proof}

\begin{figure}[H]
	\centering\footnotesize
	\begin{tikzpicture}[>=stealth,scale=7,domain=0:1,samples=99]
	\coordinate (O) at (0,0); 
	\coordinate (X) at (0.3,0); 
	\coordinate (XM) at (0.5,0);
	\coordinate (XR) at (7/8,0);
	\draw [name path=vert,dashed] (X) to ++(0,0.6); 
	\draw [name path=vertm,dashed] (XM) to ++(0,0.6); 
	\draw [name path=vertr,dashed] (XR) to ++(0,0.6); 
	
	\draw[->] (O) to ++(1.1,0); 
	\draw[->] (O) to ++(0,.6);
	
	\node at (X) {\textbullet}; 
	\node[anchor=north west] (x)  at ($(X) +(-0.05,0)$) {$(x_1,0)$};  
	\node at (XM) {\textbullet}; 
	\node[anchor=north west] (x)  at ($(XM) +(-0.05,0)$) {$(x_2,0)$}; 
	\node at (XR) {\textbullet}; 
	\node[anchor=north west] (x)  at ($(XR) +(-0.05,0)$) {$(x_3,0)$}; 
	
	\draw[name path=C] plot (\x,{sqrt(\x)*sqrt(1-\x)}) node[right,anchor=south west] {$C$};
	
	\draw[name path=ell0,->] (X) to ++(135:0.4) node[anchor=south west] {}; 
	
	\draw[name path=ell1,->] (X) to ++(45:0.7) node[anchor=east] {}; 
	
	\draw[name path=ellm0,->] (XM) to ++(135:0.7) node[anchor=south west] {}; 
	
	\draw[name path=ellm1,->] (XM) to ++(45:0.7) node[anchor=east] {}; 
	
	\draw[name path=ellr0,->] (XR) to ++(135:0.8) node[anchor=south west] {}; 
	
	\draw[name path=ellr1,->] (XR) to ++(45:0.4) node[anchor=east] {}; 
	
	\draw (X) ++(-0.05,0) arc [start angle=180, end angle=135, radius=0.05];
	\draw (X) ++(180-20:0.1) node {$\pi/4$};
	\draw (X) ++(0.05,0) arc [start angle=0, end angle=45, radius=0.05]; %
	\draw (X) ++(20:0.1) node {$\pi/4$};
	
	\draw [name intersections={of=ell0 and C, by=PF}] node at (PF) {\textbullet} node at (PF) [anchor=south] {$P_1$};
	\draw [name intersections={of=ell1 and C, by=QFI}] node at (QFI) {\textbullet} node at (QFI) [anchor=south] {$Q_1$};

\draw [name intersections={of=ellm0 and C, by=PF2}] node at (PF2) {\textbullet} node at (PF2) [anchor=south] {$P_2$};
\draw [name intersections={of=ellm1 and C, by=QF2}] node at (QF2) {\textbullet} node at (QF2) [anchor=south] {$Q_2$};

\draw [name intersections={of=ellr0 and C, by=PF3}] node at (PF3) {\textbullet} node at (PF3) [anchor=south west] {$P_3$};
\draw [name intersections={of=ellr1 and C, by=QF3}] node at (QF3) {\textbullet} node at (QF3) [anchor=south west] {$Q_3$};

\draw[name path=pfqf1, dotted] (PF) to (QFI);
\draw[name path=pfqf2, dotted] (PF2) to (QF2);
\draw[name path=pfqf3, dotted] (PF3) to (QF3);

\draw [name intersections={of=vert and pfqf1, by=H1}] node at (H1) {\textbullet} node at (H1) [anchor=east] {$H_1$};
\draw [name intersections={of=vertm and pfqf2, by=H2}] node at (H2) {\textbullet} node at (H2) [anchor=south east] {$H_2$};
\draw [name intersections={of=vertr and pfqf3, by=H3}] node at (H3) {\textbullet} node at (H3) [anchor=east] {$H_3$};

\draw[name path=q1q3, dotted] (QFI) to (QF3);
\draw [name intersections={of=q1q3 and ellm1, by=F2}] node at (F2) {\textbullet} node at (F2) [anchor=north] {$F_2$};

\draw[name path=p1p3, dotted] (PF) to (PF3);
\draw [name intersections={of=p1p3 and ellm0, by=F1}] node at (F1) {\textbullet} node at (F1) [anchor=north east] {$F_1$};

\draw[name path=h1h3, dashed, thick] (H1) to (H3);
\draw[name path=f1f3, dashed, thick] (F1) to (F2);
\end{tikzpicture}
	\caption{Intuition underlying \claimref{preserve-concavity}. }
	\label{fig:concavity-preserved1}
	\end{figure}

\begin{claim}
	\label{clm:our-induct-upper} 
	Let $U_n$ be defined by the zeros of the curve $Y = u_n \sqrt{X(1-X)}$, where 
	$u_n>0$ for all $n\geq 1$ and $X\in[0,1]$. 
	Then, we have $U_{n+1} \succcurlyeq T(U_n)$. 
\end{claim}

\begin{proof}
	Let $x_0$ be the smaller of the two roots of the equation $x_0 + u_n\sqrt{x_0(1-x_0)} = x$, and $x_1$ be the larger of the two roots of the equation $x + u_n\sqrt{x_1(1-x_1)} = x_1$. 
	So, we have
	$$x_0 = \frac{(2x+u_n^2) - u_n\sqrt{u_n^2 -4x^2 +4x}}{2(1+u_n^2)} \; ,$$
	Now, let $y_0 = u_n\sqrt{x_0(1-x_0)}$. 
	Then, we have 
	\begin{align*}
	y_0 &= u_n \frac{\sqrt{\left((2x+u_n^2) - u_n\sqrt{u_n^2 -4x^2 +4x}\right)\left((2-2x+u_n^2) + u_n\sqrt{u_n^2 -4x^2 +4x}\right)}}{2(1+u_n^2)} \; .
	\end{align*}
	
	
	We substitute $x = 1/2 - z$ and need to consider only $z\in[0,1/2]$ because the curves are symmetric around the line $X=1/2$. 
	From this substitution, we have
	\begin{align*}
	y_0 &= u_n \frac{\sqrt{\left((1+u_n^2-2z) - u_n\sqrt{u_n^2+1-4z^2}\right)\left((1+u_n^2+2z) + u_n\sqrt{u_n^2 +1-4z^2}\right)}}{2(1+u_n^2)} \; .
	\end{align*}
	
	Now, the expression of $y_1 = u_n\sqrt{x_1(1-x_1)}$ is
	\begin{align*}
	y_1 &= u_n \frac{\sqrt{\left((1+u_n^2+2z) - u_n\sqrt{u_n^2+1-4z^2}\right)\left((1+u_n^2-2z) + u_n\sqrt{u_n^2 +1-4z^2}\right)}}{2(1+u_n^2)} \; .
	\end{align*}
	
	Note that $u_{n+1} \sqrt{x(1-x)} = \frac{u_n}{\sqrt{u_n^2+1}} \sqrt{\frac14 - z^2}$. 

Now,
	\begin{align*}
	&&\mathrm{H.M.}(y_0,y_1) &\leq \frac{u_n}{\sqrt{u_n^2+1}} \sqrt{\frac14 - z^2}\\
&&\mathrm{H.M.}\left(\frac{2(1+u_n^2)}{u_n}\cdot y_0,\frac{2(1+u_n^2)}{u_n}\cdot y_1\right) &\leq \sqrt{(u_n^2+1)(1-4z^2)} \; .
	\end{align*}
	The final inequality follows from the HM-GM inequality and the following simplifications
	\begin{align*}
	\left((1+u_n^2+2z) + u_n\sqrt{u_n^2 +1-4z^2}\right)\left((1+u_n^2+2z) - u_n\sqrt{u_n^2+1-4z^2}\right) &= (u_n^2+1)(1+2z)^2 \; .\\
	\left((1+u_n^2-2z) + u_n\sqrt{u_n^2 +1-4z^2}\right)\left((1+u_n^2-2z) - u_n\sqrt{u_n^2+1-4z^2}\right) &= (u_n^2+1)(1-2z)^2 \; .
	\end{align*}
\end{proof}

The following result is used in the proof of \lemmaref{lowerbound-an}. 
\begin{lemma}
	\label{lem:inv-square-diff}
	For $x\geq 0$, we have $\frac{1}{4\left(\sqrt{x+1}-\sqrt{x}\right)^2}\leq x+\frac12$.  
\end{lemma}
\begin{proof}
	\begin{align*}
		\frac1{4\left(\sqrt{x+1}-\sqrt{x}\right)^2}&\leq x+\frac12\\
		\left(\frac{\sqrt{x+1}+\sqrt{x}}{2}\right)^2&\leq x+\frac{1}{2}=\frac{(x+1)+x}{2}\\
		\left(\frac{\sqrt{x+1}+\sqrt{x}}{2}\right)&\leq \sqrt{\frac{(\sqrt{x+1})^2+(\sqrt{x})^2}{2}}\\
		\mathrm{A.M.}\left(\sqrt{x+1},\sqrt{x}\right)&\leq  \mathrm{R.M.S.}\left(\sqrt{x+1},\sqrt{x}\right)
	\end{align*}
	The last inequality follows from the RMS-AM inequality.
\end{proof}

\begin{lemma}
	\label{lem:lowerbound-an}
	Suppose a sequence $a_1,a_2,\dots$ is given such that $a_{n+1}=2\left(\frac{\sqrt{a_n^2+1}-1}{a_n}\right)$, then 
	$$a_n\geq \frac{1}{\sqrt{\frac{1}{a_1^2}+\frac{n-1}{2}}} \; .$$
\end{lemma}
\begin{proof}
	Let $b_j\defeq\frac{1}{a_j^2}$, so $a_j=\frac{1}{\sqrt{b_j}}$. Now, it follows from $a_{j+1}=2\left(\frac{\sqrt{a_j^2+1}-1}{a_j}\right)$ that $b_{j+1}=\frac{1}{4\left(\sqrt{b_{j}+1}-\sqrt{b_j}\right)^2}$ and according to \lemmaref{inv-square-diff}, $b_{j+1}\leq b_j+\frac{1}{2}$, for $j\geq 1$. Therefore, 
	\begin{align*}
		\sum_{j=1}^{n-1}b_{j+1} &\leq \sum_{j=1}^{n-1} \left(b_j+\frac12\right) \\
		b_{n} &\leq b_1+\frac{n-1}{2}=\frac{1}{a_1^2}+\frac{n-1}{2} \; .
    \end{align*}
    Therefore, 
    \begin{align*}
		& a_n\geq \frac{1}{\sqrt{\frac{1}{a_1^2}+\frac{n-1}{2}}} \; .
	\end{align*}
\end{proof}

	\section{Applications}
\label{sec:applications}

This section discusses various consequences of \theoremref{gap-main} and other related results. 

\subsection{Distributed Coin-Tossing Protocol} 
\label{sec:fair} 

We consider constructing distributed $n$-processor coin-tossing protocols where the $i$-th processor broadcasts her message in the $i$-th round. 
We shall study this problem in the information-theoretic setting.  
Our objective is to design $n$-party distributed coin-tossing protocols where an adversary cannot bias the distribution of the final outcome significantly.

For $X_0=1/2$, one can consider the incredibly elegant ``majority protocol''~\cite{STOC:Blum83,ABCGM85,STOC:Cleve86}.
The $i$-th processor broadcasts a uniformly random bit in round $i$. 
The final outcome of the protocol is the majority of the $n$ outcomes, and an adversary can bias the final outcome by $\frac1{\sqrt{2\pi n}}$ by restarting a processor once~\cite{STOC:Cleve86}.

We construct distributed $n$-party bias-$X_0$ coin-tossing protocols, for any $X_0\in[0,1]$, and our new protocol for $X_0=1/2$ is more robust to restarting attacks than this majority protocol. 
Fix $X_0\in[0,1]$ and $n\geq 1$. 
Consider the optimal martingale $(X_0,X_1,\dotsc,X_n)$ guaranteed by \theoremref{gap-main}.
The susceptibility corresponding to any stopping time is $=C_n(X_0) \leq U_n(X_0) = \frac1{\sqrt n}\sqrt{X_0(1-X_0)}$.  
Note that one can construct an $n$-party coin-tossing protocol where the $i$-th processor broadcasts the $i$-th message, and the corresponding Doob's martingale is identical to this optimal martingale. 
An adversary who can restart a processor once biases the outcome of this protocol by at most $\frac12 C_n(X_0)$, this is discussed in \sectionref{control}. 
\begin{corollary}[Distributed Coin-tossing Protocols]
\label{corol:fair} 
For every $X_0\in[0,1]$ and $n\geq 1$ there exists an $n$-party bias-$X_0$ coin-tossing protocol such that any adversary who can restart a processor once causes the final outcome probability to deviate by $\leq \frac12C_n(X_0) \leq \frac12U_n(X_0) = \frac1{2\sqrt n} \sqrt{X_0(1-X_0)}$. 
\end{corollary}
For $X_0=1/2$, our new protocol's outcome can be changed by $\frac1{4\sqrt n}$, which is less than the $\frac1{\sqrt{2\pi n}}$ deviation of the majority protocol. 
However, we do not know whether there exists a {\em computationally efficient} algorithm implementing the coin-tossing protocols corresponding to the optimal martingales. 

\subsection{Fail-stop Attacks on Coin-tossing/Dice-rolling Protocols}
\label{sec:unfair}

A {\em two-party $n$-round bias-$X_0$ coin-tossing protocol} is an interactive protocol between two parties who send messages in alternate rounds, and $X_0$ is the probability of the coin-tossing protocol's outcome being heads. 
{\em Fair computation} ensures that even if one of the parties aborts during the execution of the protocol, the other party outputs a (randomized) heads/tails outcome. 
This requirement of guaranteed output delivery is significantly stringent, and Cleve~\cite{STOC:Cleve86} demonstrated a computationally efficient attack strategy that alters the output-distribution by $O(1/n)$, \ie, any protocol is $O(1/n)$ unfair. 
Defining fairness and constructing fair protocols for general functionalities has been a field of highly influential research~\cite{STOC:GHKL08,EC:GorKat10,C:BLOO11,TCC:AshLinRab13,TCC:Asharov14,makriyannis2014classification,TCC:ABMO15}. 
This interest stems primarily from the fact that fairness is a desirable attribute for secure-computation protocols in real-world applications. 
However, designing fair protocol even for simple functionalities like (bias-$1/2$) coin-tossing is challenging both in the two-party and the multi-party setting. 
In the multi-party setting, several works~\cite{ABCGM85,C:BeiOmrOrl10,TCC:AloOmr16} explore fair coin-tossing where the number of adversarial parties is a constant fraction of the total number of parties. 
For a small number of parties, like the two-party and the three-party setting, constructing such protocols have been extremely challenging even against computationally bounded adversaries~\cite{TCC:MorNaoSeg09,STOC:HaiTsf14,SODA:BHLT17}. 
These constructions (roughly) match Cleve's $O(1/n)$ lower-bound in the computational setting.

In the information-theoretic setting, Cleve and Impagliazzo~\cite{Cleve93martingales} exhibited that any two-party $n$-round bias-$1/2$ coin-tossing protocol are $\frac1{2560\sqrt n}$ unfair. 
In particular, their adversary is a fail-stop adversary who follows the protocol honestly except aborting prematurely. 
In the information-theoretic commitment-hybrid, there are two-party $n$-round bias-1/2 coin-tossing protocols that have $\approx 1/\sqrt n$ unfairness~\cite{STOC:Blum83,ABCGM85,STOC:Cleve86}. 
This bound matches the lower-bound of $\Omega(1/\sqrt n)$ by Cleve and Impagliazzo~\cite{Cleve93martingales}. 
It seems that it is necessary to rely on strong computational hardness assumptions or use these primitives in a non-black box manner to beat the $1/\sqrt n$ bound~\cite{TCC:DLMM11,TCC:HaiOmrZar13,TCC:DacMahMal14,BeimelHMO17}. 

We generalize the result of Cleve and Impagliazzo~\cite{Cleve93martingales} to all 2-party $n$-round bias-$X_0$ coin-tossing protocols (and improve the constants by two orders of magnitude). 
For $X_0=1/2$, our fail-stop adversary changes the final outcome probability by $\geq \frac1{24\sqrt2}\cdot \frac1{\sqrt{n+1}}$.
\begin{theorem}[Fail-stop Attacks on Coin-tossing Protocols]
\label{thm:unfairness} 
For any two-party $n$-round bias-$X_0$ coin-tossing protocol, there exists a fail-stop adversary that changes the final outcome probability of the honest party by at least $\frac1{12}C'_n(X_0)\geq \frac1{12} L'_n(X_0) \defeq\frac1{12} \sqrt{\frac2{n+1}} X_0(1-X_0)$, where $C'_1(X)\defeq X(1-X)$ and $C'_n(X)\defeq {T}^{n-1}(C'_1(X))$. 
\end{theorem}
This theorem is {\em not} a direct consequence of \theoremref{gap-main}. 
The proof relies on an entirely new inductive argument; however, the geometric technique for this recursion is similar to the proof strategy for \theoremref{gap-main}. 

Before proving the above theorem, we provide some insight into our approach. Let $\Pi = \left\langle A,B\right\rangle$ be an $n$-round bias-$X_0$ coin-tossing protocol between Alice and Bob. 
Without loss of generality, assume that Alice sends messages in rounds $1,3,\dotsc$, and Bob sends messages in rounds $2,4,\dotsc$. 
The random variable $(E_1,\dotsc,E_i)$ represents the partial transcript of the protocol at the end of round $i$. 
The random variable $X_i$ represents the expected probability of heads at the end of the protocol execution conditioned on the current partial transcript at the end of round $i$. 
Note that $(X=(X_i)_{i=0}^n,E=(E_i)_{i=1}^n )$ is a Doob's martingale. 

We construct fail-stop adversaries only. 
Suppose Alice has to send the message in round $(i+1)$ (\ie, $i$ is even), but she aborts. 
Then, the {\em defense} $D_i$ is the probability of Bob outputting heads. 
Similarly, suppose Bob is supposed to send the message in round $(i+1)$ (\ie, $i$ is odd), but he aborts. 
Then, we define $D_i$ as the probability of Alice outputting heads. 
Note that $D_i$ is $(E_1,\dotsc,E_i)$ measurable. 
In other words, the defense of round $i$ is a function only of the partial transcript at the end of that round.

The high-level idea of our construction of a good fail-stop attack is the following. 
We shall use a stopping time $\tau$ to identify appropriate partial transcripts of $\Pi$ to abort. 
Suppose we have already generated a partial transcript $(e_1,\dotsc,e_i)$ (refer~\figureref{prelim-mart-tree}), and the next messages that are possible are $e_{i+1}\in\Omega_{i+1}$. 
Suppose $\tau$ stops the martingale at $e_{i+1} = e\p j$. 
Note that $X_{i+1}=x\p j$ is the probability of heads conditioned on the transcript $\Pi$ being $(e_1,\dotsc,e_i,e_{i+1}=e\p j)$. 
Further, the defense of the other party is $D_i$. 

If $i$ is even, then Alice is supposed to send the $(i+1)$-th message. 
So, the stopping time $\tau$ is indicating Alice to abort if the message in the next round she plans to send is $e\p j$. 
Suppose $x\p j\leq D_i$. 
Then, if Alice aborts when her next message is $e\p j$, then she is increasing the probability of heads by $p\p j\abs{x\p j-D_i}$. 

So, the conclusion is the following. 
If $i$ is even and $x\p j\leq D_i$ then the advice of $\tau$ will be {\em helpful} to an adversarial Alice who is interested in increasing the probability of heads, say $A^+$. 
If $x\p j>D_i$, then the advice of $\tau$ will be helpful to an adversarial Alice who is interested in reducing the probability of heads, say $A^-$. 
Similarly, when $i$ is odd, the advice of $\tau$ is useful to either $B^+$ or $B^-$.

{\bfseries Specialized Stopping Time.} 
For this discussion, let us consider \figureref{transform-def}. 
Note that if $X_1$ is very small (that is, $X_1 < x_S(x)$) or $X_1$ is very large (that is, $X_1> x_L(x)$), then the adversary aborts. 
Furthermore, if $X_1$ is close to $X_0$  (that is, $X_1 \in [x_S(x),x_L(x)]$), then the adversary does not abort and recursively constructs the optimum stopping time. 
In particular (refer to \figureref{jump2-induct} and \figureref{induct-c1c2}) if there exists $x\p j$ and $x\p{j'}$ such that $x\p j< x_S(x)$ and $x\p{j'}>x_L(x)$ then the adversary aborts in both these two cases. 
This step is crucial to arguing that the point $Q'$ is higher than the point $Q$ in \figureref{induct-c1c2}, which, in turn, is key to the transformation definition.

However, if a stopping time stops the martingale at high as well as low values of $X_i$ then it is not evident how to to translate the susceptibility corresponding to this stopping time into output-bias achieved by a fail-stop adversary. 
So, we restrict to {\em specialized stopping times} with the following property (we use \figureref{transform-def} for reference in the following definition). 
\begin{boxedalgo}
	\label{def:special-def}
Fix $n$ and $X_0$. 
Pick any $i=n-d$ and fix $E_1=e_1,\dotsc,E_i=e_i$. 
Let $x=(X_i|E_1=e_1,\dotsc,E_i=e_i)$. 
\begin{itemize}
\item Either, the specialized stopping time stops for all $X_{i+1}< x_S(x)$ and recursively stops $X_{i+1}\geq x_S(x)$ later, or 
\item The specialized stopping time stops all $X_{i+1}>x_L(x)$ and recursively stops $X_{i+1}\leq x_L(x)$ later. 
\end{itemize}
\end{boxedalgo} 
Now, it is not evident whether specialized stopping times also have high susceptibility.

\begin{figure}
\begin{center}
\footnotesize\vfill
\begin{tikzpicture}[domain=0:1,scale=9,auto,>=stealth] 
\coordinate (O) at (0,0); 
\draw[->,name path=xaxis] (0,0) to (1.1,0) node[below] {$X$-axis}; 
\draw[->,name path=yaxis] (0,0) to (0,.5) node[above] {$Y$-axis}; 

\draw[name path=curve,densely dotted] plot (\x,{\x*(1-\x)}) node[right,anchor=south west] {$C$}; 

\coordinate (X) at (.3,0); 
\coordinate (Xl1) at (.05,0); 
\coordinate (Xr1) at (.60,0);

\node[anchor=north] (x) at (X) {$x$};  
\node[anchor=north ] (xl1) at (Xl1) {$x_\ell$};   
\node[anchor=north west] (xr1) at (Xr1) {$x_r$};

\draw[loosely dotted,-,name path=elLl1] (Xl1) to (0.05,0.4) node [right] {};

\draw[loosely dotted,-,name path=elr1] (Xr1) to (0.60,0.5) node [right] {};

\draw[densely dotted,->,name path=ell0] (X) to ++(135:0.4) node [right] {};
\draw[densely dotted,->,name path=ell1] (X) to ++(45:0.7) node [right] {}; 

\draw [name intersections={of=curve and ell0, by=P0}] node at (P0) {\textbullet};
\node[anchor=east] at (P0) {$P_1$};

\draw [name intersections={of=curve and elLl1, by=Pl1}] node at (Pl1) {\textbullet};
\node[anchor=east] at (Pl1) {$D_{\ell}$};
\draw [name intersections={of=ell0 and elLl1, by=Pl1l}] node at (Pl1l) {\textbullet};
\node[anchor=east] at (Pl1l) {$U_{\ell}$};
\draw [name intersections={of=curve and elr1, by=Pr1}] node at (Pr1) {\textbullet};
\node[anchor=west] at (Pr1) {$D_{r}$};
\draw [name intersections={of=ell1 and elr1, by=Pr1r}] node at (Pr1r) {\textbullet};
\node[anchor=west] at (Pr1r) {$U_{r}$};

\draw [name intersections={of=curve and ell1, by=P1}] node at (P1) {\textbullet};
\node[anchor=north east] at (P1) {$P_2$};

\draw[solid,name path=p0p1, red!70!black] (P0) to (P1); 
\draw[name path=vert,draw=none] (X) to +(90:0.5); 

\draw [name intersections={of=p0p1 and vert, by=Q}] node at (Q) {\textbullet};
\node [anchor=south east] at (Q) {$Q$}; 
\draw[dotted,thick, name path=xvert] (X) to (Q); 
\draw[dotted,thick,name path=qvert] (Q) to (0.3,0.5);

\draw[solid,name path=l1r1, blue!70!black] (Pl1l) to (Pr1);
\draw[solid,name path=r2l2, green!70!black] (Pr1r) to (Pl1);

\draw [name intersections={of=l1r1 and qvert, by=Qp}] node at (Qp) {$\circ$} node [anchor=south west] at (Qp) {$Q'$};
\draw [name intersections={of=r2l2 and xvert, by=Qpp}] node at (Qpp) {$\circ$} node [anchor=north west] at (Qpp) {$Q''$};

\draw[loosely dotted] (P0) to (P0 |- O); 
\node[anchor=north] at (P0 |- O) {$x_S(x)$};
\draw[loosely dotted] (P1) to (P1 |- O); 
\node[anchor=north] at (P1 |- O) {$x_L(x)$};                                                                                                                                                                                                                  \end{tikzpicture}
\end{center}
\caption{Intuition of the geometric transformation when restricted to specialized stopping times. 
The intersection of $X=x$ with lines $U_\ell D_r, P_1P_2$ and $D_\ell U_r$ are the points $Q^{'}, Q$ and $Q^{''}$ respectively. 
Note that in this figure, node $x$ has only two children $x_\ell$ and $x_r$}
\label{fig:special-intuit} 
\end{figure}

\begin{theorem}
	\label{thm:specialized-stopping}
	Let $(X_0,X_1,\dotsc,X_n)$ be a discrete-time martingale such that $X_i\in[0,1]$, for all $i\in\{1,\dotsc,n\}$, and $X_n\in\{0,1\}$. 
	Then, the following bound holds.
	$$\sup_{\text{specialized stopping time }\tau} \EX{\abs{X_\tau-X_{\tau-1}}} \geq C^{'}_n(X_0),$$
	where $C'_1(X)\defeq X(1-X)$ and $C'_n(X)\defeq {T}^{n-1}(C'_1(X))$.
\end{theorem}
Let us start with the base case $n=1$. 
Note that a specialized stopping time cannot stop the martingale at both low and high $X_1$. 
So, we consider stopping times $\tau\colon\Omega\to\{1,\dotsc,n,\infty\}$, where $\tau=\infty$ for a full transcript indicates that the adversary did not abort. 
Note that a specialized stopping time can either stop the martingale when $X_1=0$ or $X_1=1$.
In either of these two cases, the susceptibility is $C'_1(X_0) = X_0(1-X_0)$. 

For $n\geq 2$, we show that the recursive definition of the transform $T$ continues to hold even for specialized stopping time (refer \figureref{special-intuit} for intuition). 
Note that the adversary chooses the stopping time that achieves the highest susceptibility. 
So, the maximum height of $Q'$ and $Q''$ in \figureref{special-intuit} is greater than the height of $Q$. 
We emphasize that this proof crucially relies on the fact that $C'_{n-1}(X)$ lies below the curve $Y =\min\{X,1-X\}$. 
So, our result holds because $C'_1(X)$ lies below the curve $Y =\min\{X,1-X\}$.
\sectionref{special-stopping-proof} presents the full proof. 

Finally, we translate the susceptibility of a specialized stopping time into output-bias that a fail-stop adversary can enforce. 
\sectionref{attack-coin-toss} provides the full proof of \theoremref{unfairness}.

\subsubsection{Black-box Separation Results}
Gordon and Katz~\cite{EC:GorKat10} introduced the notion of {\em $1/p$-unfair secure computation} for a fine-grained study of fair computation of functionalities. 
In this terminology, \theoremref{unfairness} states that $\frac{c}{\sqrt{n+1}}X_0(1-X_0)$-unfair computation of a bias-$X_0$ coin is impossible for any positive constant $c< \frac{\sqrt 2}{12}$ and $X_0\in[0,1]$. 

Cleve and Impagliazzo's result~\cite{Cleve93martingales} states that $\frac{c}{\sqrt n}$-unfair secure computation of the bias-$1/2$ coin is impossible for any positive constant $c< \frac1{2560}$. 
This result on the hardness of computation of fair coin-tossing was translated into black-box separations results. 
These results~\cite{TCC:DLMM11,TCC:HaiOmrZar13,TCC:DacMahMal14}, intuitively, indicate that it is unlikely that $\frac{c}{\sqrt n}$-unfair secure computation of the bias-$1/2$ coin exists, for $c< \frac1{2560}$, relying solely on the black-box use of one-way functions. 
We emphasize that there are several restrictions imposed on the protocols that these works~\cite{TCC:DLMM11,TCC:HaiOmrZar13,TCC:DacMahMal14} consider; detailing all of which is beyond the scope of this draft. 
Substituting the result of \cite{Cleve93martingales} by \theoremref{unfairness}, extends the results of \cite{TCC:DLMM11,TCC:HaiOmrZar13,TCC:DacMahMal14} to general bias-$X_0$ coin-tossing protocols. 
\begin{corollary}[Informal: Black-box Separation]
\label{corol:bb-sep} 
For any $X_0\in[0,1]$ and positive constant $c< \frac{\sqrt 2}{12}$, the existence of $\frac{c}{\sqrt{n+1}}X_0(1-X_0)$-unfair computation protocol for a bias-$X_0$ coin is black-box separated from the existence of one-way functions (restricted to the classes of protocols considered by \cite{TCC:DLMM11,TCC:HaiOmrZar13,TCC:DacMahMal14}).
\end{corollary}

\subsubsection{Detailed Discussion of Our Fail-stop Attack and Proofs}
\label{sec:attack-coin-toss}

Given a stopping time $\tau$ we shall associate the following score with it 
$$S'(\tau) \defeq  \sum^{n+1}_{i=1} \E_{x \in \Omega} \bigg\vert \E[ (X_i-D_{i-1}) \mathbbm{1}_{\tau = i} \vert E_1(x),E_2(x),\ldots, E_{i-1}(x)] \bigg\vert   
$$

Intuitively, this score correctly accounts for the increase and decrease in the probability of heads in every round $i$. \footnote{
	The score is slightly pessimistic, which, we argue, is also necessary. 
	Note that our expression is of the form $\abs{\EX{(X_i-D_{i-1}) \mathbbm{1}_{\tau = i} \vert\dotsc}}$. 
	One might \naive{}ly consider using the expression $\EX{\abs{X_i-D_{i-1}} \mathbbm{1}_{\tau = i} \vert\dotsc}$ instead. 
	However, there is an issue.   
	Suppose the stopping time stops the martingale for all children of $X_i$. 
	This strategy causes the outcome to deviate by $\abs{X_i-D_i}$, and our expression correctly accounts for it (because $\EX{X_{i+1}}=X_i$). 
	However, the alternative expression accounts for it incorrectly. 
	Basically, the alternative expression might not be translatable into a deviation of outcome by a fail-stop attacker. %
}

\begin{claim}
	\label{clm:D_abort_parent_cond}
	We prove the following two statements
	\begin{itemize}
		\item If $0\leq x^{(\ell)}\leq x_0\leq x\leq 1$, (where $x_0$ is the solution of equation $x-x_0=C^{'}_d(x_0)$ in $[0,1]$), $x-D \geq \frac{2}{3} \left(x-x^{(\ell)}\right) \geq 0$, and $ x-x^{(\ell)} \geq C^{'}_d(x^{(l)})$, then, 
		$$ x-D \geq \frac{1}{3}C^{'}_{d+1}(x) \; .$$
		\item If $0\leq x\leq x_1\leq x^{(r)}\leq 1$, (where $x_1$ is the solution of equation $x_1-x=C^{'}_d(x_1)$ in $[0,1]$), $D-x\geq \frac{2}{3}\left(x^{(r)}-x\right) \geq 0$, and $\left(x^{(r)}-x\right) \geq C^{'}_d(x^{(r)})$, then $$D-x\geq \frac{1}{3}C^{'}_{d+1}(x) \; .$$
	\end{itemize}
\end{claim}
\begin{proof}
We prove the first statement. Since for each $n$, $C^{'}_n(x)=C^{'}_n(1-x)$, the second part is implied by the first part by replacing $x,D,x^{(\ell)}$ with $1-x,1-D,x^{(r)}=1-x^{(\ell)}$.

In order to show the first part, it is sufficient to show that $\frac{2}{3}(x-x^{(\ell)}) \geq \frac{1}{3} C^{'}_{d+1}(x)$.

We know that 
\begin{align*}
\frac{C^{'}_{d+1}(x)}{3} = \frac{2}{3}\cdot \frac{y_0y_1}{y_0+y_1} = \frac{2}{3}\cdot \frac{(x-x_0)(x_1-x)}{(x_1-x_0)} \; .
\end{align*}
We also know that 
\begin{align*}
&x-x^{(\ell)} \geq x-x_0 \; ,
\end{align*}
and, 
\begin{align*}
& x_1 - x_0 \geq x_1-x \; .
\end{align*}
Combining the above two relations we have 
\begin{align*}
(x-x^{(\ell)})(x_1 - x_0) &\geq (x-x_0)(x_1-x) 
\\(x-x^{(\ell)}) &\geq \frac{(x-x_0)(x_1-x)}{(x_1 - x_0)}
\\\frac{2}{3}(x-x^{(\ell)}) &\geq \frac{1}{3}\cdot \frac{2(x-x_0)(x_1-x)}{(x_1 - x_0)}
\\\frac{2}{3}(x-x^{(\ell)}) &\geq \frac{1}{3}C^{'}_{d+1}(x)
\\ \frac{2}{3}(x-x^{(\ell)}) &\geq \frac{1}{3}C^{'}_{d+1}(x)
\end{align*}
The proof follows from observing that $x-D\geq \frac{2}{3}(x-x^{(\ell)})$ from our assumption. 
\end{proof}

We will use specialized stopping time defined in \sectionref{unfair} to construct a stopping time for our fail-stop adversary. More formally, given a stopping time $\tau_1$ from \theoremref{specialized-stopping} such that $\sup_{\tau_1} \EX{\abs{X_{\tau_1}-X_{{\tau_1}-1}}} \geq C^{'}_n(X_0)$, there exists a stopping time $\tau_2$ such that $ S'(\tau_2) \geq \frac{1}{3}C^{'}_n(X_0)$. 

\begin{proof}
	The proof will proceed by induction on $n$. 
	\begin{enumerate}
		\item Base Case: For $n=1$, see \figureref{def-induct-bc}. 
		\begin{figure}[H]
    \centering
\begin{tikzpicture}\footnotesize 
 \node [ellipse split,draw,inner sep=1pt, minimum size=6pt] (v) at (1, 2) {{$ x$} \nodepart{lower} {$ D$}};
\node [ellipse,draw,inner sep=0.5pt, minimum size=6pt] (z) at (0, 1) [] {$0 $};
\node [ellipse,draw, inner sep=0.5pt, minimum size=6pt]  (x) at (2, 1) [] {$1$};
	\path
(v) edge (x)
(v) edge (z)
	;
\end{tikzpicture}
    \caption{Base Case for \theoremref{unfairness}}
    \label{fig:def-induct-bc}
\end{figure}
		Recall that $C^{'}_1(x)=x(1-x)$. We have two cases 
		\begin{itemize}
			\item If $D \geq x$, we define $\tau_2$ as the stopping time that stops only at $0$. Then, $D(1-x) \geq x(1-x) \geq \frac{1}{3}x(1-x) \geq \frac{1}{3}C^{'}_1(x)$.
			\item If $D < x$, we define $\tau_2$ as the stopping time that stops only at $1$. Then $(1-D)x \geq x(1-x) \geq \frac{1}{3}x(1-x) \geq \frac{1}{3}C^{'}_1(x)$.
		\end{itemize}
	
		\item Assume the claim is true for $n=d$, see \figureref{def-induct-indhyp}. 
		\begin{figure}[H]
    \centering
\begin{tikzpicture}[sibling distance=3.5cm, level 2/.style={sibling distance =2.5cm}]\footnotesize
 \node [ellipse split,draw,inner sep=0.5pt, minimum size=6pt]{{$x$} \nodepart{lower}{${D}$}}
child{ node[ellipse split,draw,inner sep=0.5pt, minimum size=6pt] {{$x^{(1)}$} \nodepart{lower} {${D^{(1)}}$}}
		{ node[itria] {${\frac{1}{3}C'_{d}(x^{(1)})}$} } 
	}
child{ node[ellipse split,draw,inner sep=0.5pt, minimum size=6pt] {{$x^{(2)}$} \nodepart{lower} {${D^{(2)}}$}}
	{ node[itria] {${\frac{1}{3}C'_{d}(x^{(2)})}$} } 
}
child{ node[circle, draw = none] {$\ldots$} edge from parent[draw=none]
	{ node[draw = none] {} } 
}
child{ node[ellipse split,draw,inner sep=0.5pt, minimum size=6pt] {{$x^{(k)}$} \nodepart{lower} {${D^{(k)}}$}} 
	{ node[itria] {${\frac{1}{3}C'_{d}(x^{(k)})}$} } 
}
;

\end{tikzpicture}
    \caption{Inductive Hypothesis of \theoremref{unfairness}}
    \label{fig:def-induct-indhyp}
\end{figure}
		For each edge $\left(x,x^{(j)}\right)$, if $\vert x-x^{(j)} \vert \geq C^{'}_d\left(x^{(j)}\right)$, then mark the edge. Let $$\mathtt{Marked}\defeq\{j: |x-x^{(j)}|\geq C^{'}_d(x^{(j)} \} \; .$$
		\\Without loss of generality, we assume that the nodes are \emph{in-order}. 
		
		Denote $\mathtt{Left}:=\{j:\ x^{(j)} \leq x \} \bigcap \mathtt{Marked}$ and $\mathtt{Right}:=\{j:\ x^{(j)} \geq x \}\bigcap \mathtt{Marked}$. 
		We analyze three possible cases 
		\begin{itemize}
			\item \textit{Case 1.} No edges are marked. This means that for all $j$, $\vert x-x^{(j)} \vert \leq C^{'}_d\left(x^{(j)}\right)$. 
			The adversarial strategy is to recurse on the underlying subtrees. 
			The overall deviation in this case is given by 
			$$\sum_j p^{(j)} C^{'}_d \left(x^{(j)}\right) \geq C^{'}_{d+1}\left(\sum_j p^{(j)} x^{(j)} \right) = C^{'}_{d+1}(x) \geq \frac{C^{'}_{d+1}(x)}{3} \; .$$
			
			\item  \textit{Case 2.} There exists a marked edge $j$ such that $D \leq \frac{x+2x^{(j)}}{3}$, and $x \geq x^{(j)}$, or $D \geq \frac{x+2x^{(j)}}{3}$, and $x^{(j)} \geq x$. The adversarial strategy is to abort at the parent.\\
			Suppose $D \leq \frac{x+2x^{(j)}}{3}$ and $x \geq x^{(j)}$, then $x-D \geq \frac{2}{3}(x-x^{(j)})$, the rest follows from \claimref{D_abort_parent_cond}. If $D \geq \frac{x+2x^{(j)}}{3}$ and $x^{(j)} \geq x$, then $D-x \geq \frac{2}{3}(x^{(j)}-x)$ and the rest again follows from \claimref{D_abort_parent_cond}. 
			\item  If \textit{Case 1} and \textit{Case 2} are not satisfied, then $\mathtt{Marked}$ is not empty but for any marked edge $j$ that $x\geq x^{(j)}$, we have $D>\frac{x+2x^{(j)}}{3}$ and for any marked edge $j$ that $x\leq x^{(j)}$, we have $D<\frac{x+2x^{(j)}}{3}$. Note that since $\mathtt{Marked}$ is not empty, at least one of the two sets $\mathtt{Left}$ and $\mathtt{Right}$ is not empty. Two cases can happen: 
			\begin{itemize}
				\item \textit{Case 3.1} Both $\mathtt{Left}$ and $\mathtt{Right}$ are non-empty. \\
				Then there exist $\ell^*$ and $r^*$ such that $\frac{x+2x^{(\ell^*)}}{3} < D < \frac{x+2x^{(r^*)}}{3}$ where $\ell^*:=\max_\ell \mathtt{Left}$ and $r^*:=\min_r \mathtt{Right}$. There are two sub-cases in this scenario : 
				\begin{itemize}
					\item \textit{Case 3.1.1} $\frac{x+2x^{(\ell^*)}}{3} < D \leq x$. \\
					The adversarial strategy is to follow the strategy of $\tau_1$. If the strategy of $\tau_1$ is to abort on left marked edges and recurse on the rest, then we have the following analysis: 

					For any $\ell \in \mathtt{Left}$, $\ell \leq \ell^*$, and we have 
					$$D-x^{(\ell)} > \frac{x+2x^{(\ell^*)}}{3} - x^{(\ell)} = \frac{x-x^{(\ell)} + 2(x^{(\ell^*)}-x^{(\ell)})}{3} \geq \frac{x-x^{(\ell)}}{3} > \frac{C^{'}_d(x^{(\ell)})}{3}  \; .$$

					The total deviation from aborting on the left marked edges is given by 
					\begin{align*}
					&\sum_{\ell \in \mathtt{Left}} p^{(\ell)} (D-x^{(\ell)}) \geq \sum_{\ell \in \mathtt{Left}} p^{(\ell)} \frac{x-x^{(\ell)}}{3} \geq  \sum_{\ell \in \mathtt{Left}} p^{(\ell)} \frac{C^{'}_d(x^{(\ell)})}{3} \; .
					\end{align*}
					The total deviation from recursing on the right edges and unmarked edges is given by 
				$$ \sum_{k \not\in \mathtt{Marked}} p^{(k)} \frac{C^{'}_d(x^{(k)})}{3} + \sum_{r \in \mathtt{Right}}p^{(r)}\frac{C^{'}_{d}(x^{(r)})}{3} \; , $$
				The overall deviation is 
				\begin{align*}
				&\sum_{\ell \in \mathtt{Left}} p^{(\ell)} (D-x^{(\ell)})+\sum_{r \in \mathtt{Right}}p^{(r)}\frac{C^{'}_{d}(x^{(r)})}{3} + \sum_{k \not\in \mathtt{Marked}} p^{(k)} \frac{C^{'}_d(x^{(k)})}{3} \\
				&\geq  \sum_{\ell \in \mathtt{Left}} p^{(\ell)} \frac{x-x^{(\ell)}}{3} +\sum_{r \in \mathtt{Right}}p^{(r)}\frac{C^{'}_{d}(x^{(r)})}{3} + \sum_{k \not\in \mathtt{Marked}} p^{(k)} \frac{C^{'}_d(x^{(k)})}{3} \\
				&\geq \frac{C^{'}_{d+1}(x)}{3}
				\end{align*}
			    In above, the last inequality holds due to the fact that $\tau_1$ is a specialized stopping time and martingale aborts on left marked edges and recurses on the rest which is exactly what $\tau_1$ suggests.
				\\If the strategy of $\tau_1$ is to abort on the right marked edges and recurse on the rest, then we have the following analysis.

				For any $r \in \mathtt{Right}$ such that $r \geq r^*$, we have 
				$$x^{(r)}-D>x^{(r)}-x > C^{'}_d(x^{(r)}) >\frac{C^{'}_d(x^{(r)})}{3} \; .$$
				The total deviation from aborting on the right marked edges is given by 
				\begin{align*}
				&\sum_{r \in \mathtt{Right}} p^{(r)} (x^{(r)}-D) \geq \sum_{r \in \mathtt{Right}} p^{(r)} \frac{x^{(r)}-x }{3} \geq \sum_{r \in \mathtt{Right}} p^{(r)} \frac{C^{'}_d(x^{(r)})}{3}
				\end{align*}
					The total deviation from recursing on the left edges and unmarked edges is given by 
				$$ \sum_{k \not\in \mathtt{Marked}} p^{(k)} \frac{C^{'}_d(x^{(k)})}{3} + \sum_{l \in \mathtt{Left}}p^{(l)}\frac{C^{'}_{d}(x^{(l)})}{3} \; .$$
				The overall deviation is 
				\begin{align*}
				&\sum_{r \in \mathtt{Right}} p^{(r)} (x^{(r)} -D )+\sum_{\ell \in \mathtt{Left}}p^{(l)}\frac{C^{'}_{d}(x^{(l)})}{3} + \sum_{k \not\in \mathtt{Marked}} p^{(k)} \frac{C^{'}_d(x^{(k)})}{3} \\
				&\geq  \sum_{r \in \mathtt{Right}} p^{(r)} \frac{x^{(r)} - x}{3} +\sum_{\ell \in \mathtt{Left}}p^{(\ell)}\frac{C^{'}_{d}(x^{(\ell)})}{3} + \sum_{k \not\in \mathtt{Marked}} p^{(k)} \frac{C^{'}_d(x^{(k)})}{3} \\
				&\geq \frac{C^{'}_{d+1}(x)}{3}
				\end{align*}
			    In above, the last inequality holds due to the fact that $\tau_1$ is a specialized stopping time and martingale aborts on right marked edges and recurses on the rest which is exactly what $\tau_1$ suggests.
			    
				\item \textit{Case 3.1.2.}  $x < D < \frac{x+2x^{(r^*)}}{3}$. 
					\\The adversarial strategy is the same as above : Follow the strategy of $\tau_1$. The analysis is almost identical to the one above due to symmetry. 				
				\end{itemize}
				\item \textit{Case 3.2.} Either $\mathtt{Left}$ or $\mathtt{Right}$ is empty. \\The adversarial strategy is to abort at all marked edges and recurse on all unmarked edges. \\Suppose $\mathtt{Right}$ is empty, then $\frac{x+2x^{(\ell^*)}}{3} < D < x$, where $\ell^*:=\max_\ell \mathtt{Left}$. The analysis is the same as in \textit{Case 3.1.1}. If $\mathtt{Left}$ is empty then the analysis is the same as \textit{Case 3.1.2}. \qedhere
			\end{itemize}
			
		\end{itemize}
	
	\end{enumerate}
\end{proof}

The above proof shows that $S'(\tau_2) \geq \frac{1}{3} C'_n(X_0)$. In order to estimate $C'_n(X_0)$, we define $L'_n(X)=\sqrt{\frac{2}{n+1}} X(1-X)$ and claim that $C'_n(X) \succcurlyeq L_n'(X)$. 

To prove our claim, we define the curve $G'_n(X)\defeq a'_nX(1-X)$ such that $a'_1=1$ and 
$a'_{n+1}= 2\left(\frac{\sqrt{a'^2_n+1}-1}{a'_n}\right)$ for $n\geq 1$ and we prove by induction that $C'_n \succcurlyeq G'_n$ for all $n$ as below: (analogous to the one shown for \lemmaref{gap-lower-upper}) 

{\bfseries Base Case of $n=1$.} 
Since, $C'_1(X)=G'_1(X) = X(1-X)$, it is obvious that $C'_1\succcurlyeq G'_1$. 

{\bfseries Inductive Argument.} 
Suppose we have $C'_n\succcurlyeq G'_n$. 
Then, we have $T(C'_n)\succcurlyeq T(G'_n)$ (by~\claimref{above-to-above}).
Note that $C'_{n+1}=T(C'_n)$. 
We know that $T(G'_n)\succcurlyeq G'_{n+1}$ (refer to \claimref{our-induct} ). 
Consequently, it follows that $C'_{n+1}\succcurlyeq G'_{n+1}$. 

So far, we have proved that $C'_n \succcurlyeq G'_n$ for all $n$. Recall that $G'_n(X)\defeq a'_nX(1-X)$ such that $a'_1=1$ and 
$a'_{n+1}= 2\left(\frac{\sqrt{a'^2_n+1}-1}{a'_n}\right)$. Now, by using \lemmaref{lowerbound-an}, we conclude that $a'_n\geq \sqrt{\frac{2}{n+1}}$. Thus, $C'_{n} \succcurlyeq L'_{n}$.
Now we can say that $S'(\tau_2) \geq \frac{1}{3} C'_n(X_0) \geq \frac{1}{3} L'_n(X_0)$. Further, any contribution to this score is attributable to one of the following four interactions: (1) \tuple{A^+,B} (\ie, adversarial Alice increasing the probability of heads by aborting), 
(2) \tuple{A^-,B}, 
(3) \tuple{A,B^+}, and 
(4) \tuple{A,B^-}. By an averaging argument, this implies that one of the parties can deviate the outcome of the other party by at least $\frac{1}{12}L'_n(X_0)$. This concludes our proof of \theoremref{unfairness}.

Similar to the previous section, \theoremref{unfairness} extends to $\omega$-faceted dice-rolling protocols by considering any subset $S\subseteq\{0,1,\dotsc,\omega-1\}$ of outcomes, and considering the final outcome being in $S$ as the interesting event for the martingale. 
 
\subsubsection{Discussion of Specialized Stopping Time - Proof of \theoremref{specialized-stopping}}
\label{sec:special-stopping-proof}

Before proving the theorem, we define the sequence of functions $\{g_n\}_{n=1}^\infty$ recursively. 
Let $A_n(X_0)$ 
be the set of 
all martingales $X=\left(X_0,X_1,\dotsc,X_{n}\right)$ 
such that for each $0\leq i\leq n-1$, $X_i\in[0,1]$ and $X_n\in\{0,1\}$. We define $$g_1(X_0):=\inf_{X\in A_1(X_0)}\sup_{\tau\in \mathcal{F}_1(X_0,X_1)}\E\vert X_{\tau}-X_{\tau-1} \vert$$ where $\mathcal{F}_1(X_0,X_1):=\{\tau_1, \tau_2\}$ and $\tau_1$ is an stopping time defined on martingale $(X_0,X_1)$ such that $\tau_1(X_0,X_1)=1$ if $X_1=0$ and $\tau_1(X_0,X_1)=\infty$ if $X_1=1$; and $\tau_2(X_0,X_1)=1$ if $X_1=1$ and $\tau_2(X_0,X_1)=\infty$ if $X_1=0$. Note that $\mathcal{F}_1(X_0,X_1)$ represents the set of all specialized stopping times in martingale $(X_0,X_1)$.
$A_1(X_0)$ consists of only one martingale and $\E\vert X_{\tau_1}-X_{\tau_1-1}\vert=\E\vert X_{\tau_2}-X_{\tau_2-1}\vert=X_0(1-X_0)$ which implies that $g_1(X_0)=X_0(1-X_0)$. 
We define 
$$g_n(X_0):=\inf_{X\in A_n(X_0)}\sup_{\tau\in \mathcal{F}_n(X_0,X_1,\dots,X_n)}\E\vert X_{\tau}-X_{\tau-1} \vert$$
where $\mathcal{F}_n(X_0,X_1,\dots,X_n)$ denotes the set of all specialized stopping times like $\tau$ defined on martingale $X=(X_0,X_1,\dots,X_n)$ which could be one of the following two cases:

Suppose $X_0=x$ and $X_1=x^{*}$. Then, let $x_0\in[0,1]$ be the solution of equation $x-x_0= g_{n-1}(x_0)$ and $x_1\in[0,1]$ be the solution of equation $x_1-x= g_{n-1}(x_1)$. 

\begin{enumerate}
	\item For all $x^{*}\leq x_0$, $\tau(x,x^{*},X_2,\dots,X_n)=1$ and for all $x^{*}> x_0$, $\tau(x,x^{*},X_2,\dots,X_n)=1+\tau^{'}(x^{*},X_2,\dots,X_n)$ for some $\tau^{'}\in \mathcal{F}_{n-1}(x^{*},X_2,\dots,X_n)$. This corresponds to
	the case that the specialized stopping time stops for all $x^{*}\leq x_0$ and recursively stops for all $x^{*}\geq x_0$ later.
	\item For all $x^{*}\geq x_1$, $\tau(x,x^{*},X_2,\dots,X_n)=1$ and for all $x^{*}< x_1$, $\tau(x,x^{*},X_2,\dots,X_n)=1+\tau^{'}(x^{*},X_2,\dots,X_n)$ for some $\tau^{'}\in \mathcal{F}_{n-1}(x^{*},X_2,\dots,X_n)$. This corresponds to
	the case that the specialized stopping time stops for all $x^{*}\geq x_1$ and recursively stops for all $x^{*}\leq x_1$ later.
\end{enumerate}      

To prove \theoremref{specialized-stopping}, it suffices to prove the following claim. 
\begin{claim}
	\label{clm:specialized-main-claim}
	Let $C'_1(x)=x(1-x)$ and the 
	curve $C'_n$ is achieved by applying 
	transformation $T$ on the curve $C'_{n-1}$ i.e. $C'_n=T\left(C'_{n-1}\right)$. Then, we have $g_n(x)=C'_n(x)$ for any $x\in[0,1]$.
\end{claim}
We first describe the intuitive idea behind the proof and then give a technical proof afterwards. \newline 
{\bfseries Proof Sketch.}
We use induction on $n$ to prove the claim. For $n=1$ and for each $x\in [0,1]$, we have $g_1(x)=x(1-x)=C'_1(x)$. 
Now, we assume that for each $x\in [0,1]$, $g_n(x)=C'_n(x)$.
Since $C'_{n+1}(X)=T\left(C'_n(X)\right)$, it suffices to prove that $g_{n+1}(X)=T\left(g_n(X)\right)$ because it implies that $g_{n+1}(X)= T\left(g_n(X)\right)=T\left(C'_n(X)\right)=C'_{n+1}(X)$. 
Let us consider martingale $(X_0,X_1,\dots,X_n, X_{n+1})$ where $X_0=x$ and $X_1\in \{x^{(1)},\dots,x^{(t)}\}$. According to the induction hypothesis, the adversary is guaranteed to get $g_n(x^{(j)})=C'_n(x^{(j)})$ as the score in any martingale of depth $n$ if she chooses an appropriate stopping time in $\mathcal{F}_n(x^{(j)},X_2,\dots,X_{n+1})$. 

We define left marked edges as the set $\{j: x^{(j)} \leq x, \text{ and, } |x-x^{(j)}|\geq C^{'}_{n}(x^{(j)}) \}$ and right marked edges as the set $\{j: x^{(j)} \geq x, \text{ and, } |x-x^{(j)}|\geq C^{'}_{n}(x^{(j)}) \}$.
Now, to prove that $g_{n+1}(X) = T\left(g_n(X)\right)$ it suffices to show that in any arbitrary martingale in $A_{n+1}(x)$, the maximum score that could be achieved by either stopping the martingale at only left marked edges at time $1$ or stopping the martingale at only right marked edges at time $1$, is always guaranteed to be greater than or equal to $T(g_n(x))=T(C'_n(x))=C'_{n+1}(x)$.   
In \figureref{special-intuit}, we are considering a martingale $(x,X_1,\dots,X_n,X_{n+1})$ such that $X_1$ can take only two values either $x_{l}$ or $x_{r}$ with probabilities $p_l$ and $p_r$ respectively. Note that $x_{l} \leq x_S(x)$ and $x_L(x) \leq x_{r}$. Any specialized stopping time $\tau$ either stops at $x_l$ and continues at $x_r$ or stops at $x_r$ and continues at $x_l$. Here, the curve $C'_{n}$ represents the points $(x,g_n(x))$ for $0\leq x\leq 1$.
According to the induction hypothesis, in martingale $(x_l, X_2, \dots, X_n)$, the score $g_n(x_l)$ is guaranteed to be achieved (so the contribution of score when martingale doesn't stop at this edge is $p_lg_n(x_l)$) but if martingale stops at time $1$ at edge $(x,x_l)$, then the contribution of score for this edge is $p_l|x-x_l|$. A similar thing can be said about $x_r$. 
We can observe that while the point $Q^{''}$ (which is the intersection of line $D_lU_r$ with line $X=x$ and its height corresponds to the score achieved when martingale stops at $x_r$ and continues at $x_l$) lies below the point $Q=(x,g_{n+1}(x))$ (which is the intersection of line $P_1P_2$ with line $X=x$ and its height corresponds to $T(g_n(x))$), the point $Q^{'}$ (which is the intersection of line $U_lD_r$ with line $X=x$ and its height corresponds to the score achieved by stopping martingale at $x_l$ and allowing it to continue at $x_r$) is above the point $Q$. Observe that the maximum of the two scores achieved in these two strategies is always greater than or equal to $T(g_n(x))$. Moreover, if $x_S$ is chosen as $x_l$
and $x_L$ is chosen as $x_r$, then $Q=Q^{'}=Q^{''}$ and the value $T(C'_{n})(x)$ can be achieved for some martingale. This means that $g_{n+1}(x)=T(g_n(x))=T(C'_n(x))=C'_{n+1}(x)$ for any $x\in [0,1]$.

\subsection{Influencing Discrete Control Processes} 
\label{sec:control}

Lichtenstein~\etal~\cite{LLS1989} considered the problem of an adversary influencing the outcome of a stochastic process through mild interventions. 
For example, an adversary attempts to bias the outcome of a distributed $n$-processor coin-tossing protocol, where, in the $i$-th round, the processor $i$ broadcasts her message. 
This model is also used to characterize randomness sources that are adversarially influenced, for example,~\cite{FOCS:SriZuc94,SODA:KenRabSin96,zuckerman1996simulating,nisan1996extracting,nisan1999extracting,FOCS:TreVad00,FOCS:DodSpe02,FOCS:DOPS04,TCC:DodPiePrz06,TCC:BosDod07}.

Consider the sample space $\Omega=\Omega_1\times\Omega_2\times\dotsi\times\Omega_n$ and a joint distribution $(E_1,\dotsc,E_n)$ over the sample space. 
We have a function $f\colon\Omega\to\zo$ such that $\EX{f(E_1,\dotsc,E_n)}=X_0$. 
This function represents the protocol that determines the final outcome from the public transcript. 
The filtration, at time-step $i$, reveals the value of the random variable $E_i$ to the adversary. 
We consider the corresponding Doob's martingale $(X_0,X_1,\dotsc,X_n)$. 
Intuitively, $X_i$ represents the probability of $f(E_1,\dotsc,E_n)=1$ conditioned on the revealed values $(E_1=e_1,\dotsc,E_i=e_i)$. 
The adversary is allowed to intervene only once. 
She can choose to intervene at time-step $i$, reject the current sample $E_i=e_i$, and substitute it with a fresh sample from $E_i$. 
This intervention is identical to {\em restarting} the $i$-th processor if the adversary does not like her message. 
Note that this intervention changes the final outcome by 
  $$(X_{i-1}|E_1=e_1,\dotsc,E_{i-1}=e_{i-1}) - (X_i|E_1=e_1,\dotsc,E_i=e_i)$$

We shall use a stopping time $\tau$ to represent the time-step where an adversary decides to intervene. 
However, for some $(E_1=e_1,\dotsc,E_n=e_n)$ the adversary may not choose to intervene. 
Consequently, we consider stopping times $\tau\colon\Omega\to\{1,\dotsc,n,\infty\}$, where the stopping time being $\infty$ corresponds to the event that the adversary did not choose to intervene. 
In the Doob martingale discussed above, as a direct consequence of \theoremref{gap-main}, there exists a stopping time $\tau^*$ with susceptibility $\geq C_n(X_0)$. 
Note that susceptibility measures the expected (unsigned) magnitude of the deviation, if an adversary intervenes at $\tau^*$.  
Some of these contributions to susceptibility shall increase the probability of the final outcome being 1, and the remaining shall decrease the probability of the final outcome being 1. 
By an averaging argument, there exists a stopping time $\tau\colon\Omega\to\{1,\dotsc,n,\infty\}$ that biases the outcome of $f$ by at least $\geq\frac12 C_n(X_0)$, whence the following corollary. 
\begin{corollary}[Influencing Discrete Control Processes]
\label{corol:control}
Let $\Omega_1,\dotsc,\Omega_n$ be arbitrary sets, and $(E_1,\dotsc,E_n)$ be a joint distribution over the set $\Omega\defeq\Omega_1\times\dotsi\times\Omega_n$. 
Let $f\colon\Omega\to\zo$ be a function such that $\probX{f(E_1,\dotsc,E_n)=1}=X_0$. 
Then, there exists an adversarial strategy of intervening once to bias the probability of the outcome away from $X_0$ by $\geq \frac12C_n(X_0) \geq \frac12L_n(X_0) =\frac{1}{\sqrt{2n-1}}X_0(1-X_0)$. 
\end{corollary}
The previous result of \cite{Cleve93martingales} applies only to $X_0=1/2$ and they ensure a deviation of $1/320\sqrt n$. 
For $X_0=1/2$, our result ensures a deviation of (roughly) $1/4\sqrt{2n} \approx 1/5.66\sqrt n$.

\subsubsection{Influencing Multi-faceted Dice-rolls} 

\corollaryref{control} generalizes to the setting where $f\colon\Omega\to\{0,1,\dotsc,\omega-1\}$, \ie, the function $f$ outputs an arbitrary $\omega$-faceted dice roll. 
In fact, we quantify the deviation in the probability of any subset $S\subseteq\{0,1,\dotsc,\omega-1\}$ of outcomes caused by an adversary intervening once.  
\begin{corollary}[Influencing Multi-faceted Dice-Rolls]
\label{corol:control2}
Let $\Omega_1,\dotsc,\Omega_n$ be arbitrary sets, and $(E_1,\dotsc,E_n)$ be a joint distribution over the set $\Omega\defeq\Omega_1\times\dotsi\times\Omega_n$. 
Let $f\colon\Omega\to\{0,1,\dotsc,\omega-1\}$ be a function with $\omega\geq 2$ outcomes, $S\subseteq\{0,1,\dotsc,\omega-1\}$ be any subset of outcomes, and  $\probX{f(E_1,\dotsc,E_n)\in S}=X_0$. 
Then, there exists an adversarial strategy of intervening once to bias the probability of the outcome being in $S$ away from $X_0$ by $\geq \frac12C_n(X_0) \geq \frac12L_n(X_0) =\frac{1}{\sqrt{2n-1}}X_0(1-X_0)$. 
\end{corollary}
\corollaryref{control} and \corollaryref{control2} are equivalent to each other. 
Clearly \corollaryref{control} is a special case of \corollaryref{control2}. 
\corollaryref{control2}, in turn, follows from \corollaryref{control} by considering ``$f(E_1,\dotsc,E_n)\in S$'' as the interesting event for the martingale. 
We state these two results separately for conceptual clarity and ease of comparison with the prior work.

\subsection{$L_2$ Gaps and their Tightness}
\label{sec:l2-gap} 

Finally, to demonstrate the versatility of our geometric approach, we measure large $L_2$-norm gaps in martingales.  
\begin{figure}[H]
	\begin{boxedminipage}{\linewidth}
		\begin{minipage}{0.5\linewidth}
			\input{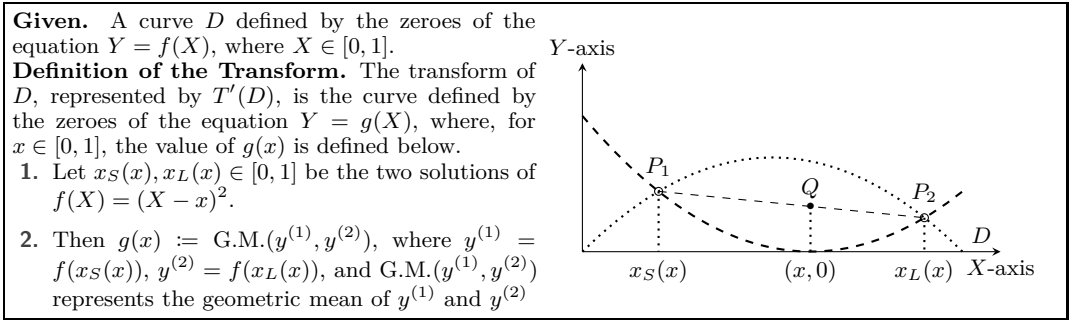} 
		\end{minipage}
		\begin{minipage}{0.5\linewidth}
\footnotesize\vfill
\begin{tikzpicture}[domain=0:1,scale=5,auto,>=stealth] 
\coordinate (O) at (0,0); 
\draw[->,name path=xaxis] (0,0) to (1.1,0) node[below] {$X$-axis}; 
\draw[->,name path=yaxis] (0,0) to (0,.5) node[above] {$Y$-axis}; 

\draw[dotted,name path=curve,thick] plot (\x,{\x*(1-\x)}) node[right,anchor=south west] {$D$}; 

\coordinate (X) at (.6,0); 

\node[anchor=north] (x) at (X) {$(x,0)$};  

\draw[dashed,name path=curve1,thick] plot (\x,{(\x-0.6)^2}) node[right,anchor=south west] {};
%
\draw [name intersections={of=curve and curve1}] (intersection-2) circle (0.3pt)
coordinate (P1)
;
\node at (P1) [above = 1mm of P1] {$P_2$};
\draw [name intersections={of=curve and curve1}] (intersection-1) circle (0.3pt)
coordinate (P0)
;
\node at (P0) [above = 1mm of P0] {$P_1$};

\draw[dashed,name path=p0p1] (P0) to (P1); 
\draw[name path=vert,draw=none] (X) to +(90:0.5); 
\draw [name intersections={of=p0p1 and vert, by=Q}] node at (Q) {\textbullet};
\node [anchor=south] at (Q) {$Q$}; 
\draw[dotted,thick] (X) to (Q); 

\draw[dotted,thick] (P0) to (P0 |- O); 
\node[anchor=north] at (P0 |- O) {$x_S(x)$};
\draw[dotted,thick] (P1) to (P1 |- O); 
\node[anchor=north] at (P1 |- O) {$x_L(x)$};

%
\end{tikzpicture}\vfill
		\end{minipage}
	\end{boxedminipage}
	\caption{Definition of transform of a curve $D$, represented by $T'(D)$. 
		The locus of the point $Q$ (in the right figure) defines the curve $T'(D)$.
	}
	\label{fig:transform-def-l2}
\end{figure}
\begin{theorem}
\label{thm:gap-main-2} 
Let $(X_0,X_1,\dotsc,X_n)$ be a discrete-time martingale such that 
  $X_n\in\{0,1\}$. 
Then, the following bound holds.
    $$\sup_{\text{stopping time }\tau} \EX{\left({X_\tau-X_{\tau-1}}\right)^2} \geq D_n(X_0)\defeq \frac1n X_0(1-X_0)$$
Furthermore, for all $n\geq1$ and $X_0\in[0,1]$, there exists a martingale $(X_0,\dotsc,X_n)$ such that for any stopping time $\tau$, it has $\EX{\left({X_\tau-X_{\tau-1}}\right)^2} = D_n(X_0)$. 
\end{theorem}

\begin{proof}
    
We shall proceed by induction on $n$. 

\noindent{\bfseries Base Case $n=1$.} 
Note that in this case (see \figureref{transform-induct-bc}) the optimal stopping time is $\tau=1$. 
  $$\mathrm{opt}_1(X_0,2) = D_1(X_0) = (1-X_0)X_0^2 + X_0(1-X_0)^2 = X_0(1-X_0) \; .$$

\noindent{\bfseries General Inductive Step.} 
Let us fix $X_0=x$ and $n=d\geq 2$. 
We proceed analogous to the argument in \sectionref{large-gap-l1}. 
The adversary can either decide to stop at the child $j$ (see \figureref{jump2-induct} for reference) or continue to the subtree rooted at it to find a better stopping time.
 
\begin{figure}[H]
\begin{center}\footnotesize
\begin{tikzpicture}[domain=0:1,scale=5,auto,>=stealth] \footnotesize
\coordinate (O) at (0,0); 
\draw[->,name path=xaxis] (0,0) to (1.1,0) node[anchor = north] {$X$-axis}; 
\draw[->,name path=yaxis] (0,0) to (0,.7) node[anchor = south] {$Y$-axis} coordinate (ytop); 
\draw[->,name path=ypaxis,draw=none] (1,0) to (1,.75) coordinate (yptop); 
\draw[->,name path=ypaxis1,draw=none] (0.85,0) to (0.85,.65) coordinate (yptop1);

\coordinate (X) at (.3,0); 
\node at (X) {\textbullet}; 
\node[anchor=north west] (x)  at ($(X) +(-0.05,0)$) {$X=(x,0)$};  
\draw[dashed, name path=xvert] (X) to ++(90:0.6) coordinate (top); 

\draw[dotted,name path=curve,thick] plot (\x,{\x*(1-\x)}) node[anchor=south west] {$D_d$}; 
\draw [name intersections={of=curve and xvert, by=m}] coordinate (mid) at (m);
\draw[dashed,name path=curve1,thick] plot (\x,{(\x-0.3)^2}) node[right,anchor=south west] {};

\draw [name intersections={of=curve and curve1}] (intersection-2) circle (0.3pt)
coordinate (P1)
;
\node at (P1) [below = .5mm of P1] {$P_2$};
\draw [name intersections={of=curve and curve1}] (intersection-1) circle (0.3pt)
coordinate (P0)
;
\node at (P0) [below = 0mm of P0] {$P_1$};
%
%
\draw [name intersections={of=yaxis and curve1, by=L}] node at (L) {\textbullet};
\node[anchor=east] at (L) {};
%
%
\draw [name intersections={of=ypaxis and curve1, by=R}] node at (R) {\textbullet};
\node[anchor=west] at (R) {};
\draw [name intersections={of=ypaxis1 and curve1, by=R1}] node at (R1) {};
\node[anchor=west] at (R1) {};
%
%
%

\begin{pgfonlayer}{background}
\draw [fill=gray!50!,draw=none] (ytop) rectangle (1,0); 
\begin{scope}
\draw[draw=none,fill=white] plot (\x,{\x*(1-\x)});
\draw[draw=white,fill=white] plot (\x,{(\x-0.3)^2}) to (1,0) to (O);
\end{scope} 

\begin{scope}
\clip (P0) rectangle (mid); 
\draw[thick] plot (\x,{\x*(1-\x)});
\end{scope}

\begin{scope}
\clip (P1) rectangle (top); 
\draw[thick] plot (\x,{\x*(1-\x)});
\end{scope}

\begin{scope}
\clip (P0) rectangle (ytop); 
\draw[thick] plot (\x,{(\x-0.3)^2});
\end{scope}

\begin{scope}
\clip (P1) rectangle (ytop-|R); 
\draw[thick] plot (\x,{(\x-0.3)^2});
\end{scope}

\end{pgfonlayer}
%

\foreach \x/\y/\i/\where in {0.14/0.15/1/south east, 
                             0.6/0.25/2/south west, 
                             0.7/0.4/3/west, 
                             0.5/0.45/4/south, 
                             0.15/0.3/5/east, 
                             0.5/.35/6/east,
                             0.2/0.25/7/south west}{
  \node (x\i) at (\x,\y) {$\circ$}; \node[anchor=\where] at (x\i) {$Z^{(\i)}$};
}
\draw[dashed,pattern=dots, pattern color=gray] (x1.center) to (x2.center) to (x3.center) to (x4.center) to (x5.center) to cycle; 

\draw[name path=temp,draw=none] (x1.center) to (x2.center); 
\draw [name intersections={of=temp and xvert, by=Qp}] node at (Qp) {$\circ$} node [anchor=west,inner sep=1mm] at (Qp) {$Q'$};

\draw[dashed,name path=p0p1] (P0) to (P1); 
\draw[name path=vert,draw=none] (X) to +(90:0.5); 
\draw [name intersections={of=p0p1 and vert, by=Q}] node at (Q) {\textbullet};
\node [anchor=north east,inner sep=0] at (Q) {$Q$}; 
\draw[dotted,thick] (X) to (Q); 

\end{tikzpicture}
\end{center}
\caption{Intuitive Summary of the inductive argument. 
  Our objective is to pick the set of points $\{Z\p1,Z\p2\dotsc\}$ in the gray region to minimize the length of the intercept $XQ'$ of their (lower) convex hull on the line $X=x$. 
  Clearly, the unique optimal solution corresponds to including both $P_1$ and $P_2$ in this set.%
}
\label{fig:induct-l2}
\end{figure}

Overall, the adversary gets the following contribution from the $j$-th child 
  $$\max\left\{(x\p j-x)^2 , D_{d-1}(x\p j)\right\}$$
The adversary obtains a score that is at least the height of $Q$ in \figureref{induct-l2}. 
Further, a martingale designer can choose $t=2$, and $Z\p 1=P_1$ and $Z\p 2=P_2$ to define the optimal martingale. 
Similar to \theoremref{gap-main}, the scores corresponding to all possible stopping times in the optimal martingale are identical. 

We can argue that the height of $Q$ is the geometric-mean of the heights of $P_1$ and $P_2$. 
This observation defines the geometric transformation $T'$ in \figureref{transform-def-l2}.
For this transformation, we demonstrate that $D_n(X_0)=\frac1nX_0(1-X_0)$ is the solution to the recursion $D_n = {T'}^{n-1}(D_1)$ in \claimref{Characterising-Second-Norm}. 
\end{proof}

\begin{claim}
	\label{clm:Characterising-Second-Norm}
	Let $D_1$ be the curve defined as the zeros of the equation $Y=X(1-X)$ and for $n>1$, $D_n$ is obtained as applying the transformation $T'$, defined in \figureref{transform-def-l2}, to the curve $D_{n-1}$. We claim that for each $x\in [0,1]$, $D_n(x)=d_nx(1-x)$ where $d_n=\frac1n$. 
	\end{claim}

\begin{proof}
	We use induction on $n$ to prove that for each $x\in [0,1]$, we have $D_n(x)=d_nx(1-x)$ where $d_n=\frac1n$. Base case $n=1$, is obvious. Now, assuming that $D_n(x)=d_nx(1-x)$ where $d_n=\frac1n$, we will prove that $D_{n+1}(x)=d_{n+1}x(1-x)$ where $d_{n+1}=\frac1{n+1}$.
	Let's fix $x\in [0,1]$ and let $x_0$ and $x_1$ be respectively the smaller and larger root of the equation $d_nx^*(1-x^*)=(x-x^*)^2$. Then we have
	
	$$
	x_0=\frac{(2x+d_n)-\sqrt{d_n^2+4d_nx(1-x)}}{2(1+d_n)} \; ,
	$$
	$$
	x_1=\frac{(2x+d_n)+\sqrt{d_n^2+4d_nx(1-x)}}{2(1+d_n)} \; .
	$$
	Let $y_0=d_nx_0(1-x_0)$ and $y_1=d_nx_1(1-x_1)$, then we have the following relations: $$y_0=d_n\cdot\frac{(2x+d_n)-\sqrt{d_n^2+4d_nx(1-x)}}{2(1+d_n)}\cdot \frac{(2(1-x)+d_n)+\sqrt{d_n^2+4d_nx(1-x)}}{2(1+d_n)}$$
	
	$$y_1=d_n\cdot\frac{(2x+d_n)+\sqrt{d_n^2+4d_nx(1-x)}}{2(1+d_n)}\cdot\frac{(2(1-x)+d_n)-\sqrt{d_n^2+4d_nx(1-x)}}{2(1+d_n)}$$
	
	Now, according to the definition of transformation $T'$ in \figureref{transform-def-l2}, we have $D_{n+1}(x)=\sqrt{y_0y_1}$ and:
	
	\begin{align*}
		\sqrt{y_0y_1}&=\frac{d_n}{4(1+d_n)^2}\cdot\sqrt{\left(\left(2x+d_n\right)^2-\left(d_n^2+4d_nx(1-x)\right)\right)\left(\left(2(1-x)+d_n\right)^2-\left(d_n^2+4d_nx(1-x)\right)\right)} \\
		&=\frac{d_n}{4(1+d_n)^2}\cdot \sqrt{\left(4x^2(1+d_n)\right)\left(4(1-x)^2(1+d_n)\right)}\\
		&=\frac{d_n}{1+d_n}x(1-x)=\frac{\frac{1}{n}}{1+\frac{1}{n}}x(1-x)
		\\&=\frac{1}{n+1}x(1-x)
	\end{align*}
	\end{proof}  

Note that, for any martingale $(X_0,\dotsc,X_n)$ with $X_n\in\zo$, we have $\EX{\sum_{i=1}^n (X_i-X_{i-1})^2} = \EX{X_n^2-X_0^2} = X_0(1-X_0)$. 
Therefore, by an averaging argument, there exists a round $i$ such that $\EX{(X_i-X_{i-1})^2} \geq \frac1nX_0(1-X_0)$. 
\theoremref{gap-main-2} proves the existence of a martingale that achieves the lower-bound even for non-constant stopping times. 
This result provides a technique to obtain the upper-bound to $C_n(X)$ in \lemmaref{gap-lower-upper}. 
 
\subsection{Alternate Proof for $U_{n+1} \succcurlyeq T(U_{n})$ }
\label{app:alt-proof}
\begin{proof}
	Recall that we defined $D_n$ as the zeros of the curve $Y = \frac{1}{n}X(1-X)$. 
	Since $U_n$ is defined by the zeros of the curve $Y = \sqrt{\frac{1}{n}X(1-X)}$, by squaring the $Y$-values for $U_n$, we can obtain the curve $D_n$. 
	This is illustrated in \figureref{un-dn-transform}. 
	Denote points on curve $U_n$ as $P_1 := (x_0,y_0)$, $P_2 := (x_1,y_1)$ and points on curve $D_n$ as $P'_1 := (x_0,y'_0)$, $P'_2 := (x_1,y'_1)$. 
	In the left-hand figure, let $\alpha := x - x_0$ and $\beta := x_1-x$, then $y_0 = \alpha $ and $y_1 = \beta$ and $Q = H.M. (\alpha,\beta)$. After squaring, in the right-hand figure $Q = (H.M.(\alpha,\beta))^2$. Note by definition of the transformation $T'$, we have that $Q' =G.M.(\alpha^2,\beta^2)$. We show that $G.M.(\alpha^2,\beta^2) \geq (H.M.(\alpha,\beta))^2$ as follows 
	\begin{align*}
	&& G.M.(\alpha^2,\beta^2) &\geq \left(H.M.(\alpha,\beta)\right)^2\\
	&& \left(G.M.(\alpha^2,\beta^2)\right)^{1/2} &\geq H.M.(\alpha,\beta)\\
	&& G.M.(\alpha,\beta) &\geq H.M.(\alpha,\beta),
	\end{align*}
	which is true by the standard $G.M. \geq H.M.$ inequality. 
	Now recall that the locus of the point $Q'$ defines the curve $T'(D_n) = D_{n+1}$ (From \claimref{Characterising-Second-Norm})
	and we know that $D_{n+1} = U^2_{n+1}$. Also, after squaring the $Y$-axis, the locus of the point $Q$ defines the curve $T^2(U_n)$, therefore we have just shown that $U^2_{n+1} \succcurlyeq T^2(U_{n})$, which means that $U_{n+1} \succcurlyeq T(U_{n})$.  
\end{proof}

Note that, for any martingale $(X_0,\dotsc,X_n)$ with $X_n\in\zo$, we have $\EX{\sum_{i=1}^n (X_i-X_{i-1})^2} = \EX{X_n^2-X_0^2} = X_0(1-X_0)$. 
Therefore, by an averaging argument, there exists a round $i$ such that $\EX{(X_i-X_{i-1})^2} \geq \frac1nX_0(1-X_0)$. 
\theoremref{gap-main-2} proves the existence of a martingale that achieves the lower-bound even for non-constant stopping times. 

This result provides an alternate technique to obtain the upper-bound to $C_n(X)$ in \lemmaref{gap-lower-upper}.

\begin{figure}[H]
	\begin{boxedminipage}{\linewidth}
		\begin{minipage}{0.5\linewidth}
\footnotesize\vfill
\begin{tikzpicture}[domain=0:1,scale=5,auto,>=stealth] 

\coordinate (O) at (0,0); 
\draw[->,name path=xaxis] (0,0) to (1.1,0) node[below] {$X$-axis}; 
\draw[->,name path=yaxis] (0,0) to (0,.5) node[above] {$Y$-axis};

\draw[dotted,name path=curve,thick] plot (\x,{sqrt(\x*(1-\x))}) node[right,anchor=south west] {$U_n$}; 

\coordinate (X) at (.3,0); 

\node[anchor=north] (x) at (X) {$(x,0)$};  

\draw[dashed,->,name path=ell0] (X) to ++(135:0.4) node [right] {$\ell_1$};
\draw[dashed,->,name path=ell1] (X) to ++(45:0.7) node [right] {$\ell_2$}; 

\draw [name intersections={of=curve and ell0, by=P0}] node at (P0) {\textbullet};
\node[anchor=east] at (P0) {$P_1$};
\draw[dotted] (P0) |- (0,0);
\node[anchor=north] at (P0 |- O) {$(x_0,0)$};

\draw [name intersections={of=curve and ell1, by=P1}] node at (P1) {\textbullet};
\node[anchor=west] at (P1) {$P_2$};
\draw[dotted] (P1) |- (0,0); 
\node[anchor=north] at (P1 |- O) {$(x_1,0)$};

\draw[dashed,name path=p0p1] (P0) to (P1); 
\draw[name path=vert,draw=none] (X) to +(90:0.5); 
\draw [name intersections={of=p0p1 and vert, by=Q}] node at (Q) {\textbullet};
\node [anchor=south] at (Q) {$Q$}; 
\draw[dotted,thick] (X) to (Q); 


\end{tikzpicture}\vfill
		\end{minipage}
		\begin{minipage}{0.5\linewidth}
\footnotesize\vfill
\begin{tikzpicture}[domain=0:1,scale=5,auto,>=stealth] 
\coordinate (O) at (0,0); 
\draw[->,name path=xaxis] (0,0) to (1.1,0) node[below] {$X$-axis}; 
\draw[->,name path=yaxis] (0,0) to (0,.5) node[above] {$Y$-axis}; 

\draw[dotted,name path=curve,thick] plot (\x,{\x*(1-\x)}) node[right,anchor=south west] {$D_n$}; 

\coordinate (X) at (.3,0); 

\node[anchor=north] (x) at (X) {$(x,0)$};  

\draw[dashed,name path=curve1,thick] plot (\x,{(\x-0.3)^2}) node[right,anchor=south west] {};
%
\draw [name intersections={of=curve and curve1}] (intersection-2) circle (0.2pt)
coordinate (P1)
;
 \node at (P1) [above = 1mm of P1] {$P'_2$};
\draw [name intersections={of=curve and curve1}] (intersection-1) circle (0.2pt)
coordinate (P0)
;
 \node at (P0) [above = 1mm of P0] {$P'_1$};

\draw[dashed,name path=p0p1] (P0) to (P1); 
\draw[name path=vert,draw=none] (X) to +(90:0.5); 
\draw [name intersections={of=p0p1 and vert}] (intersection-1) circle (0.2pt)
coordinate (Q)
;
\node [anchor=south] at (Q) {$Q'$}; 
\draw[dotted,thick] (X) to (Q); 
\draw[dotted] (P0) |- (0,0);
\node[anchor=north] at (P0 |- O) {$(x_0,0)$};
\draw[dotted] (P1) |- (0,0);
\node[anchor=north] at (P1 |- O) {$(x_1,0)$};


 \coordinate (Qog) at (0.3, 441/5000);
 \draw[fill=black] (Qog) circle (0.2pt) node [anchor=north east] {$Q$};

\end{tikzpicture}\vfill
		\end{minipage}
	\end{boxedminipage}
	\caption{Intuitive summary of the Proof of \claimref{our-induct-upper}. In the left-hand figure, we have $U_n$ and the locus of the point $Q$ defines the curve $T(U_n)$. Recall that $D_n$ is defined by the zeros of the curve $Y = \frac{1}{n}X(1-X)$. 
		Intuitively we can say that $D_n = (U_n)^2$. By squaring the $Y$-axis in the left-hand figure, we get the right-hand figure. Since $D_{n+1}=T^{'}(D_n)$ (From \claimref{Characterising-Second-Norm})
		, and the locus of the point $Q'$ defines this curve, we only need to show that $Q'$ is always above $Q$ in the right-hand figure in order to prove our original claim. 
	}
	\label{fig:un-dn-transform}
\end{figure}
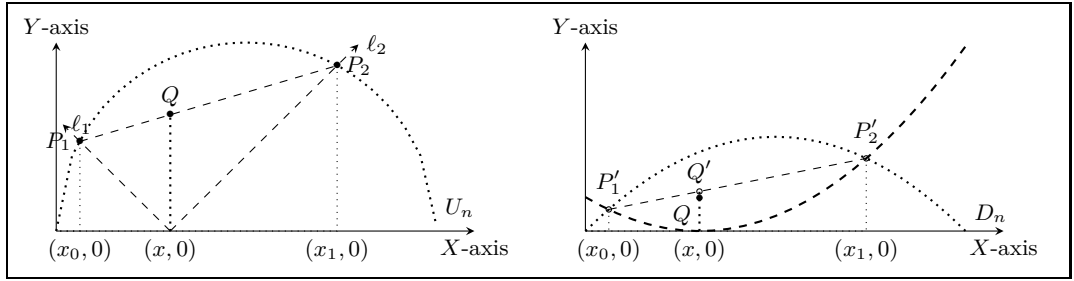

	\newpage 

\begin{thebibliography}{10}
		
		\bibitem{TCC:AloOmr16}
		Bar Alon and Eran Omri.
		\newblock Almost-optimally fair multiparty coin-tossing with nearly
		three-quarters malicious.
		\newblock In Martin Hirt and Adam~D. Smith, editors, {\em TCC~2016-B: 14th
			Theory of Cryptography Conference, Part~I}, volume 9985 of {\em Lecture Notes
			in Computer Science}, pages 307--335, Beijing, China,
		October~31~--~November~3, 2016. Springer, Heidelberg, Germany.
		\newblock \href {http://dx.doi.org/10.1007/978-3-662-53641-4_13}
		{\path{doi:10.1007/978-3-662-53641-4_13}}.
		
		\bibitem{TCC:Asharov14}
		Gilad Asharov.
		\newblock Towards characterizing complete fairness in secure two-party
		computation.
		\newblock In Yehuda Lindell, editor, {\em TCC~2014: 11th Theory of Cryptography
			Conference}, volume 8349 of {\em Lecture Notes in Computer Science}, pages
		291--316, San Diego, CA, USA, February~24--26, 2014. Springer, Heidelberg,
		Germany.
		\newblock \href {http://dx.doi.org/10.1007/978-3-642-54242-8_13}
		{\path{doi:10.1007/978-3-642-54242-8_13}}.
		
		\bibitem{TCC:ABMO15}
		Gilad Asharov, Amos Beimel, Nikolaos Makriyannis, and Eran Omri.
		\newblock Complete characterization of fairness in secure two-party computation
		of {Boolean} functions.
		\newblock In Yevgeniy Dodis and Jesper~Buus Nielsen, editors, {\em TCC~2015:
			12th Theory of Cryptography Conference, Part~I}, volume 9014 of {\em Lecture
			Notes in Computer Science}, pages 199--228, Warsaw, Poland, March~23--25,
		2015. Springer, Heidelberg, Germany.
		\newblock \href {http://dx.doi.org/10.1007/978-3-662-46494-6_10}
		{\path{doi:10.1007/978-3-662-46494-6_10}}.
		
		\bibitem{TCC:AshLinRab13}
		Gilad Asharov, Yehuda Lindell, and Tal Rabin.
		\newblock A full characterization of functions that imply fair coin tossing and
		ramifications to fairness.
		\newblock In Amit Sahai, editor, {\em TCC~2013: 10th Theory of Cryptography
			Conference}, volume 7785 of {\em Lecture Notes in Computer Science}, pages
		243--262, Tokyo, Japan, March~3--6, 2013. Springer, Heidelberg, Germany.
		\newblock \href {http://dx.doi.org/10.1007/978-3-642-36594-2_14}
		{\path{doi:10.1007/978-3-642-36594-2_14}}.
		
		\bibitem{ABCGM85}
		Baruch Awerbuch, Manuel Blum, Benny Chor, Shafi Goldwasser, and Silvio Micali.
		\newblock How to implement bracha's {O}(log n) byzantine agreement algorithm.
		\newblock {\em Unpublished manuscript}, 1985.
		
		\bibitem{Azuma1967}
		Kazuoki Azuma.
		\newblock Weighted sums of certain dependent random variables.
		\newblock {\em Tohoku Math. J. (2)}, 19(3):357--367, 1967.
		\newblock URL: \url{https://doi.org/10.2748/tmj/1178243286}, \href
		{http://dx.doi.org/10.2748/tmj/1178243286}
		{\path{doi:10.2748/tmj/1178243286}}.
		
		\bibitem{BeimelHMO17}
		Amos Beimel, Iftach Haitner, Nikolaos Makriyannis, and Eran Omri.
		\newblock Tighter bounds on multi-party coin flipping via augmented weak
		martingales and differentially private sampling.
		\newblock In {\em 2018 IEEE 59th Annual Symposium on Foundations of Computer
			Science (FOCS)}, pages 838--849. IEEE, 2018.
		
		\bibitem{C:BLOO11}
		Amos Beimel, Yehuda Lindell, Eran Omri, and Ilan Orlov.
		\newblock {$1/p$}-{Secure} multiparty computation without honest majority and
		the best of both worlds.
		\newblock In Phillip Rogaway, editor, {\em Advances in Cryptology --
			{CRYPTO}~2011}, volume 6841 of {\em Lecture Notes in Computer Science}, pages
		277--296, Santa Barbara, CA, USA, August~14--18, 2011. Springer, Heidelberg,
		Germany.
		\newblock \href {http://dx.doi.org/10.1007/978-3-642-22792-9_16}
		{\path{doi:10.1007/978-3-642-22792-9_16}}.
		
		\bibitem{C:BeiOmrOrl10}
		Amos Beimel, Eran Omri, and Ilan Orlov.
		\newblock Protocols for multiparty coin toss with dishonest majority.
		\newblock In Tal Rabin, editor, {\em Advances in Cryptology -- {CRYPTO}~2010},
		volume 6223 of {\em Lecture Notes in Computer Science}, pages 538--557, Santa
		Barbara, CA, USA, August~15--19, 2010. Springer, Heidelberg, Germany.
		\newblock \href {http://dx.doi.org/10.1007/978-3-642-14623-7_29}
		{\path{doi:10.1007/978-3-642-14623-7_29}}.
		
		\bibitem{STOC:Blum83}
		Manuel Blum.
		\newblock How to exchange (secret) keys (extended abstract).
		\newblock In {\em 15th Annual {ACM} Symposium on Theory of Computing}, pages
		440--447, Boston, MA, USA, April~25--27, 1983. {ACM} Press.
		\newblock \href {http://dx.doi.org/10.1145/800061.808775}
		{\path{doi:10.1145/800061.808775}}.
		
		\bibitem{TCC:BosDod07}
		Carl Bosley and Yevgeniy Dodis.
		\newblock Does privacy require true randomness?
		\newblock In Salil~P. Vadhan, editor, {\em TCC~2007: 4th Theory of Cryptography
			Conference}, volume 4392 of {\em Lecture Notes in Computer Science}, pages
		1--20, Amsterdam, The Netherlands, February~21--24, 2007. Springer,
		Heidelberg, Germany.
		\newblock \href {http://dx.doi.org/10.1007/978-3-540-70936-7_1}
		{\path{doi:10.1007/978-3-540-70936-7_1}}.
		
		\bibitem{SODA:BHLT17}
		Niv Buchbinder, Iftach Haitner, Nissan Levi, and Eliad Tsfadia.
		\newblock Fair coin flipping: Tighter analysis and the many-party case.
		\newblock In Philip~N. Klein, editor, {\em 28th Annual {ACM}-{SIAM} Symposium
			on Discrete Algorithms}, pages 2580--2600, Barcelona, Spain, January~16--19,
		2017. {ACM-SIAM}.
		\newblock \href {http://dx.doi.org/10.1137/1.9781611974782.170}
		{\path{doi:10.1137/1.9781611974782.170}}.
		
		\bibitem{STOC:Cleve86}
		Richard Cleve.
		\newblock Limits on the security of coin flips when half the processors are
		faulty (extended abstract).
		\newblock In {\em 18th Annual {ACM} Symposium on Theory of Computing}, pages
		364--369, Berkeley, CA, USA, May~28--30, 1986. {ACM} Press.
		\newblock \href {http://dx.doi.org/10.1145/12130.12168}
		{\path{doi:10.1145/12130.12168}}.
		
		\bibitem{Cleve93martingales}
		Richard Cleve and Russell Impagliazzo.
		\newblock Martingales, collective coin flipping and discrete control processes
		(extended abstract), 1993.
		
		\bibitem{TCC:DLMM11}
		Dana {Dachman-Soled}, Yehuda Lindell, Mohammad Mahmoody, and Tal Malkin.
		\newblock On the black-box complexity of optimally-fair coin tossing.
		\newblock In Yuval Ishai, editor, {\em TCC~2011: 8th Theory of Cryptography
			Conference}, volume 6597 of {\em Lecture Notes in Computer Science}, pages
		450--467, Providence, RI, USA, March~28--30, 2011. Springer, Heidelberg,
		Germany.
		\newblock \href {http://dx.doi.org/10.1007/978-3-642-19571-6_27}
		{\path{doi:10.1007/978-3-642-19571-6_27}}.
		
		\bibitem{TCC:DacMahMal14}
		Dana {Dachman-Soled}, Mohammad Mahmoody, and Tal Malkin.
		\newblock Can optimally-fair coin tossing be based on one-way functions?
		\newblock In Yehuda Lindell, editor, {\em TCC~2014: 11th Theory of Cryptography
			Conference}, volume 8349 of {\em Lecture Notes in Computer Science}, pages
		217--239, San Diego, CA, USA, February~24--26, 2014. Springer, Heidelberg,
		Germany.
		\newblock \href {http://dx.doi.org/10.1007/978-3-642-54242-8_10}
		{\path{doi:10.1007/978-3-642-54242-8_10}}.
		
		\bibitem{FOCS:DOPS04}
		Yevgeniy Dodis, Shien~Jin Ong, Manoj Prabhakaran, and Amit Sahai.
		\newblock On the (im)possibility of cryptography with imperfect randomness.
		\newblock In {\em 45th Annual Symposium on Foundations of Computer Science},
		pages 196--205, Rome, Italy, October~17--19, 2004. {IEEE} Computer Society
		Press.
		\newblock \href {http://dx.doi.org/10.1109/FOCS.2004.44}
		{\path{doi:10.1109/FOCS.2004.44}}.
		
		\bibitem{TCC:DodPiePrz06}
		Yevgeniy Dodis, Krzysztof Pietrzak, and Bartosz Przydatek.
		\newblock Separating sources for encryption and secret sharing.
		\newblock In Shai Halevi and Tal Rabin, editors, {\em TCC~2006: 3rd Theory of
			Cryptography Conference}, volume 3876 of {\em Lecture Notes in Computer
			Science}, pages 601--616, New York, NY, USA, March~4--7, 2006. Springer,
		Heidelberg, Germany.
		\newblock \href {http://dx.doi.org/10.1007/11681878_31}
		{\path{doi:10.1007/11681878_31}}.
		
		\bibitem{FOCS:DodSpe02}
		Yevgeniy Dodis and Joel Spencer.
		\newblock On the (non)universality of the one-time pad.
		\newblock In {\em 43rd Annual Symposium on Foundations of Computer Science},
		pages 376--387, Vancouver, British Columbia, Canada, November~16--19, 2002.
		{IEEE} Computer Society Press.
		\newblock \href {http://dx.doi.org/10.1109/SFCS.2002.1181962}
		{\path{doi:10.1109/SFCS.2002.1181962}}.
		
		\bibitem{ICALP:GolKalPar15}
		Shafi Goldwasser, Yael~Tauman Kalai, and Sunoo Park.
		\newblock Adaptively secure coin-flipping, revisited.
		\newblock In Magn{\'u}s~M. Halld{\'o}rsson, Kazuo Iwama, Naoki Kobayashi, and
		Bettina Speckmann, editors, {\em ICALP 2015: 42nd International Colloquium on
			Automata, Languages and Programming, Part~II}, volume 9135 of {\em Lecture
			Notes in Computer Science}, pages 663--674, Kyoto, Japan, July~6--10, 2015.
		Springer, Heidelberg, Germany.
		\newblock \href {http://dx.doi.org/10.1007/978-3-662-47666-6_53}
		{\path{doi:10.1007/978-3-662-47666-6_53}}.
		
		\bibitem{STOC:GHKL08}
		S.~Dov Gordon, Carmit Hazay, Jonathan Katz, and Yehuda Lindell.
		\newblock Complete fairness in secure two-party computation.
		\newblock In Richard~E. Ladner and Cynthia Dwork, editors, {\em 40th Annual
			{ACM} Symposium on Theory of Computing}, pages 413--422, Victoria, British
		Columbia, Canada, May~17--20, 2008. {ACM} Press.
		\newblock \href {http://dx.doi.org/10.1145/1374376.1374436}
		{\path{doi:10.1145/1374376.1374436}}.
		
		\bibitem{EC:GorKat10}
		S.~Dov Gordon and Jonathan Katz.
		\newblock Partial fairness in secure two-party computation.
		\newblock In Henri Gilbert, editor, {\em Advances in Cryptology --
			{EUROCRYPT}~2010}, volume 6110 of {\em Lecture Notes in Computer Science},
		pages 157--176, French Riviera, May~30~--~June~3, 2010. Springer, Heidelberg,
		Germany.
		\newblock \href {http://dx.doi.org/10.1007/978-3-642-13190-5_8}
		{\path{doi:10.1007/978-3-642-13190-5_8}}.
		
		\bibitem{TCC:HaiOmrZar13}
		Iftach Haitner, Eran Omri, and Hila Zarosim.
		\newblock Limits on the usefulness of random oracles.
		\newblock In Amit Sahai, editor, {\em TCC~2013: 10th Theory of Cryptography
			Conference}, volume 7785 of {\em Lecture Notes in Computer Science}, pages
		437--456, Tokyo, Japan, March~3--6, 2013. Springer, Heidelberg, Germany.
		\newblock \href {http://dx.doi.org/10.1007/978-3-642-36594-2_25}
		{\path{doi:10.1007/978-3-642-36594-2_25}}.
		
		\bibitem{STOC:HaiTsf14}
		Iftach Haitner and Eliad Tsfadia.
		\newblock An almost-optimally fair three-party coin-flipping protocol.
		\newblock In David~B. Shmoys, editor, {\em 46th Annual {ACM} Symposium on
			Theory of Computing}, pages 408--416, New York, NY, USA, May~31~--~June~3,
		2014. {ACM} Press.
		\newblock \href {http://dx.doi.org/10.1145/2591796.2591842}
		{\path{doi:10.1145/2591796.2591842}}.
		
		\bibitem{Hoeffding63}
		Wassily Hoeffding.
		\newblock Probability inequalities for sums of bounded random variables.
		\newblock {\em Journal of the American Statistical Association},
		58(301):13--30, 1963.
		\newblock URL:
		\url{https://www.tandfonline.com/doi/abs/10.1080/01621459.1963.10500830},
		\href
		{http://arxiv.org/abs/https://www.tandfonline.com/doi/pdf/10.1080/01621459.1963.10500830}
		{\path{arXiv:https://www.tandfonline.com/doi/pdf/10.1080/01621459.1963.10500830}},
		\href {http://dx.doi.org/10.1080/01621459.1963.10500830}
		{\path{doi:10.1080/01621459.1963.10500830}}.
		
		\bibitem{SODA:KenRabSin96}
		Claire Kenyon, Yuval Rabani, and Alistair Sinclair.
		\newblock Biased random walks, lyapunov functions, and stochastic analysis of
		best fit bin packing (preliminary version).
		\newblock In {\'E}va Tardos, editor, {\em 7th Annual {ACM}-{SIAM} Symposium on
			Discrete Algorithms}, pages 351--358, Atlanta, Georgia, USA, January~28--30,
		1996. {ACM-SIAM}.
		
		\bibitem{LLS1989}
		David Lichtenstein, Nathan Linial, and Michael Saks.
		\newblock Some extremal problems arising from discrete control processes.
		\newblock {\em Combinatorica}, 9(3):269--287, 1989.
		
		\bibitem{makriyannis2014classification}
		Nikolaos Makriyannis.
		\newblock On the classification of finite boolean functions up to fairness.
		\newblock In {\em International Conference on Security and Cryptography for
			Networks}, pages 135--154. Springer, 2014.
		
		\bibitem{TCC:MorNaoSeg09}
		Tal Moran, Moni Naor, and Gil Segev.
		\newblock An optimally fair coin toss.
		\newblock In Omer Reingold, editor, {\em TCC~2009: 6th Theory of Cryptography
			Conference}, volume 5444 of {\em Lecture Notes in Computer Science}, pages
		1--18. Springer, Heidelberg, Germany, March~15--17, 2009.
		\newblock \href {http://dx.doi.org/10.1007/978-3-642-00457-5_1}
		{\path{doi:10.1007/978-3-642-00457-5_1}}.
		
		\bibitem{nisan1996extracting}
		Noam Nisan.
		\newblock Extracting randomness: how and why-a survey.
		\newblock In {\em ccc}, page~44. IEEE, 1996.
		
		\bibitem{nisan1999extracting}
		Noam Nisan and Amnon Ta-Shma.
		\newblock Extracting randomness: A survey and new constructions.
		\newblock {\em J. Comput. Syst. Sci.}, 58(1):148--173, 1999.
		
		\bibitem{FOCS:SriZuc94}
		Aravind Srinivasan and David Zuckerman.
		\newblock Computing with very weak random sources.
		\newblock In {\em 35th Annual Symposium on Foundations of Computer Science},
		pages 264--275, Santa Fe, New Mexico, November~20--22, 1994. {IEEE} Computer
		Society Press.
		\newblock \href {http://dx.doi.org/10.1109/SFCS.1994.365688}
		{\path{doi:10.1109/SFCS.1994.365688}}.
		
		\bibitem{FOCS:TreVad00}
		Luca Trevisan and Salil~P. Vadhan.
		\newblock Extracting randomness from samplable distributions.
		\newblock In {\em 41st Annual Symposium on Foundations of Computer Science},
		pages 32--42, Redondo Beach, CA, USA, November~12--14, 2000. {IEEE} Computer
		Society Press.
		\newblock \href {http://dx.doi.org/10.1109/SFCS.2000.892063}
		{\path{doi:10.1109/SFCS.2000.892063}}.
		
		\bibitem{zuckerman1996simulating}
		David Zuckerman.
		\newblock Simulating bpp using a general weak random source.
		\newblock {\em Algorithmica}, 16(4-5):367--391, 1996.
		
	\end{thebibliography}

	\addcontentsline{toc}{section}{References}

	\appendix
	\section{More Technical Proof of \theoremref{gap-main}}
\label{app:geometric-proof-l1}

\begin{claim}
	For each $d\geq 1$, let $C_d$ denote a curve over $X\in[0,1]$ that includes all points $(x,\mathrm{opt}_d(x,1))$. Let $T$ be the transformation defined in \figureref{transform-def}.
	Then, $C_{d}=T^{d-1}(C_1)$ where $T^{0}$ denotes the identity transformation and $T^{k}$ denotes the transformation achieved by composing $T$ with itself $k$ times. Moreover, for every depth $d$, there exists a martingale of bias $x$ whose $\mathrm{max\textnormal{-}score}$ in $L_1$-norm is equal to $\mathrm{opt}_d(x,1)$ and for each $i=1,\dots,n$, $\abs{\Omega_i}=2$.  
\end{claim}
\begin{proof}
	Let $C_1$ be the curve defined by
	the zeros of the equation $Y=2X(1-X)$ such that 
	$Y\geq 0$. Let $C_{d+1}$ be the curve obtained by 
	applying the transformation $T$ on $C_d$.\\
	Let $X'=\{X'=\{X'_i\}_{i=0}^{n},E'=\{E'_i\}_{i=1}^
	{n}\}$ be a martingale over the sample space 
	$\Omega'=\Omega'_1\times\dots\times\Omega'_n$ such
	that for each $i\in[n]$, $\Omega'_i=\{0,1\}, 
	E'_i(0)=l, E'_i(1)=r, X'_0=x$, for each 
	$(e_1,e_2,\dots,e_{n-1})\in \{l,r\}^{n-1}$, 
	$X'_n(e_1,\dots,e_{n-1},l)=0$ and 
	$X'_n(e_1,\dots,e_{n-1},r)=1$, for each $i\in 
	\{1,\dotsc,n-1\}$, $X'_i(e_1,e_2,\dots,e_{i-1},l)$ is the $X$
	coordinate of the interception of the line 
	$Y=-X+X'_{i-1}(e_1,\dots,e_{i-1})$ and the curve 
	$C_{n-i}$, and $X'_i(e_1,e_2,\dots,e_{i-1},r)$ is 
	the $X$ coordinate of the interception of the line
	$Y=X-X'_{i-1}(e_1,\dots,e_{i-1})$ and the curve 
	$C_{n-i}$.\\ Moreover, for each 
	$(e_1,e_2,\dots,e_{i-1})\in \{l,r\}^{i-1}$,  
	$$\Pr[l|E'_1=e_1,\dots,E'_{i-1}=e_{i-1}]=\frac{X'_
		{i}(e_1,\dots,e_{i-1},r)-X'_{i-1}(e_1,\dots,e_{i-1
		})}{X'_{i}(e_1,\dots,e_{i-1},r)-X'_{i}(e_1,\dots,e
		_{i-1},l)}$$ and
	$$\Pr[r|E'_1=e_1,\dots,E'_{i-1}=e_{i-1}]=\frac{X'_
		{i-1}(e_1,\dots,e_{i-1})-X'_{i}(e_1,\dots,e_{i-1},
		l)}{X'_{i}(e_1,\dots,e_{i-1},r)-X'_{i}(e_1,\dots,e
		_{i-1},l)}$$
	We claim that for each martingale 
	$\{X=\{X_i\}_{i=1}^n,E=\{E_i\}_{i=1}^n\}$ with 
	respect to the sample space 
	$\Omega=\Omega_1\times\dots\times\Omega_n$, we 
	have $\mathrm{max\textnormal{-}score}_1(X,E)\geq 
	\mathrm{max\textnormal{-}score}_1(X',E')$ and so 
	$\mathrm{opt}_d(x,1)=\mathrm{max\textnormal{-}sc
		ore}_1(X',E')$.\\ We prove our claim by induction on
	the depth of the martingale i.e. $n$. \\For the base
	case $n=1$, suppose $\Omega_1=\{1,\dots,t\}$, 
	$E_1(j)=e_1^{(j)}$. Without loss of generality, we 
	assume that $X_1(e_1^{(1)})\leq X_1(e_1^{(2)})\dots\leq 
	X_1(e_1^{(t)})$. Then there exist $p^{(1)},\dots,p^{(n)}$ such 
	that $x=\sum_{j=1}^{t}p^{(j)}X_1(e_1^{(j)})$ and 
	$\sum_{j=1}^{t}p^{(j)}=1$. In this case, since 
	$X_1(e_1^{(j)})$ is $0$ or $1$, there exists some $s$ 
	such that $x=\sum_{j=s+1}^{t}p^{(j)}$ and 
	$$\mathrm{max\textnormal{-}score}_1(X,E)=(p^{(1)}+\dots+
	p^{(s)})x+(p^{(s+1)}+\dots+p^{(t)})(1-x)=2x(1-x)$$ 
	But, the maximum score of a martingale $(X,E)$ 
	with respect to $\Omega_1=\{0,1\}$ such that 
	$\Pr[0]=x$ and $\Pr[1]=1-x$ is also 
	$2x(1-x)$.\\ Suppose that the claim is true for 
	depth $d$, and we want to prove it for the depth 
	$d+1$. \\
	Suppose that the martingale $\{ 
	X=\{X_i\}_{i=0}^{d+1},E=\{E_i\}^{d+1}_{i=1} \}$ 
	over $\Omega=\Omega_1\times 
	\Omega_2\times\dots\times\Omega_{d+1}$ is given such that $\Omega_1=\{1,\dots,t\}$ and for each $j\in \Omega_1$,
	$E_1(j)=e_1^{(j)}$.
	

	Note that for each $j\in \{1,2,\dotsc,t\}$, we define the 
	martingale 
	$\{V^{(j)}=\{V^{(j)}_i\}_{i=0}^{d}, 
	E^{(j)}=\{E^{(j)}_i\}_{i=2}^{d+1}\}$ over 
	$\Omega_2\times\dots\times\Omega_{d+1}$ where 
	$V^{(j)}_{i}(e_2,\dots,e_{d+1}):=X_{i+1}(e_1^{(j)},e_2,
	\dots,e_{d+1})$ is a martingale of depth $d$. Observe that
	for 
	any $j$
	and any value of $V_0=X_1(e_1^{(j)})$, there exists an 
	stopping time 
	$\tau^{(j)}_{max}(V^{(j)},E^{(j)}):\Omega_2\times
	\dots\times\Omega_{d+1}\rightarrow\{2,\dots,n\}$ 
	that maximizes the score of the martingale 
	$(V^{(j)},E^{(j)})$. Now, note that 
	$\tau_{max}(X,E)(e_1^{(j)},e_2,\dots,e_{d+1})$ equals $1$ 
	(which means that the martingale stops at time 
	$1$) when 
	$$|X_0-X_1(e_1^{(j)})|\geq 
	\mathrm{max\textnormal{-}score}_1(V^{(j)},E^{(j)})$$ 
	or equals $\tau^{(j)}_{max}(V^{(j)},E^{(j)})(e_2,
	\dots,e_{d+1})$ when $$|X_0-X_1(e_1^{(j)})|\leq 
	\mathrm{max\textnormal{-}score}_1(V^{(j)},E^{(j)})$$
	Let us define 
	$B_j\defeq\mathrm{max}(\mathrm{max\textnormal{-}score}_1(
	V^{(j)},E^{(j)}),|X_0-X_1(e_1^{(j)})|)$.
	We represent each point $Z^{(j)}\eqdef(X_1(e_1^{(j)}),B_j)$ in a 
	plane, see \figureref{induct}. In this plane, for each point $(x,y)$, 
	the value $y$ represents the score of a stopping time in a martingale 
	whose average is $x$ (the first value that the martingale takes).
	Since $X$ is a martingale, we have $X_0=\sum_{j=1}^{t}p^{(j)}X_1(e_1^{(j)})$.
	It also follows from the definition of $B_j$ that
	$$\mathrm{max\textnormal{-}score}_1(X,E)=\sum_{j=1}^{t}\Pr[E_1=j]B_j=\sum_{j=1}^{t}p^{(j)}B_j.$$
	Therefore, we have $$(X_0,\mathrm{max\textnormal{-}score}_1(X,E
	))=\sum_{j=1}^{t}p^{(j)}(X_1(e_1^{(j)}),B_j)=\sum_{j=1}^{t}p^{(j)}Z^{(j)}.$$
	Consequently, the point $(X_0,\mathrm{max\textnormal{-}score}_1(X,E
	))$ lies on the intersection of the line $X=X_0$ and the 
	convex hull of the points $Z^{(1)},\dotsc, Z^{(t)}$ 
	(Note that the argument is true even if we assume that $t$ is not finite).     
	
	It follows from the inductive hypothesis that for each 
	$j$, there exists a martingale of depth
	$d$, $\{X'^{(j)}=\{X'^{(j)}_i\}_{i=1}^{d+1}, 
	E'^{(j)}=\{E'^{(j)}_i\}_{i=2}^{d+1}\}$ over 
	$\Omega_2^{'}\times\dots\times\Omega_{d+1}^{'}$ such that $X'^{(j)}_1=X_1(e_1^{(j)})$ and for each $i\in \{2,\dots,d+1\}$, $\vert \Omega_i^{'}\vert=2$ and
	$\mathrm{max\textnormal{-}score}_1(X',E')=\mathrm{opt}_d(X'^{(j)}_1,1)=\mathrm{opt}_d(X_1(e_1^{(j)}))$. Therefore, $\mathrm{max\textnormal{-}score}_1(V^{(j)},E^{(j)})\geq \mathrm{opt}_d(X_1(e_1^{(j)}),1)$. This implies that
	$$ B_j\geq \mathrm{max}(\mathrm{opt}_d(X_1(e_1^{(j)}),1),|X_0-X_1(e_1^{(j)})|) $$
	that means the points $Z^{(1)},\dotsc, Z^{(t)}$ lie above 
	the curve defined by the zeros of the 
	equation $Y=\mathrm{max}(\mathrm{opt}_{d}(X,1),|X_0-
	X|)=\mathrm{max}(C_d(X),
	|X-X_0|)$. 
	Note that according to the inductive hypothesis, $C_d$ (the zeros of the equation 
	$Y=\mathrm{opt}_{d}(X,1)$) is equal to the curve $T^{d-1}(C_1)$ which is concave downward as a consequence of \claimref{preserve-concavity}. Thus, the intersection of the line $X=X_0$ and the 
	convex hull of the points $Z^{(1)},\dotsc, Z^{(t)}$ is above the point $Q=(x,C_{d+1}(x))$, see \figureref{induct}. Moreover, by choosing $t=2$ and $Z^{(1)}=P_1$ and $Z^{(2)}=P_2$, the score $T(C_d)(x)$ (point $Q$) can be achieved. But, note that according to the inductive hypothesis, the points $P_1$ and $P_2$ can be achieved by the martingale $\{X'^{(j)}=\{X'^{(j)}_i\}_{i=1}^{d+1}, E'^{(j)}=\{E'^{(j)}_i\}_{i=2}^{d+1}\}$. This shows
	that the martingale of depth $d+1$ with optimal score is achieved
	when for each $i\in \{1,\dotsc,d+1\}$, $|\Omega_i|=2$. Also as mentioned earlier, the height of the point $Q$ is $T(C_d)(x)$ and according to induction hypothesis, $C_d(x)=T^{d-1}(C_1)(x)$, so $\mathrm{opt}_{d+1}(x,1)=T^{d}(C_1)(x)$.    
\end{proof}

\end{document}